\documentclass[12pt]{article}
\usepackage[margin=1in]{geometry}

\usepackage[english]{babel}
\usepackage[latin1]{inputenc}
\usepackage{amssymb}
\usepackage{amsbsy}
\usepackage{amsmath}
\usepackage{graphicx}
\usepackage{float}
\usepackage{hyperref}
\hypersetup{
    colorlinks=true,
    linkcolor=blue,
    filecolor=magenta,      
    urlcolor=blue,
    citecolor=blue
}
\usepackage{mathtools}

\usepackage[toc,page]{appendix}
\usepackage{makecell}

\usepackage{natbib}
\usepackage{tikz}
\usetikzlibrary{positioning}

\usepackage{algorithm}
\usepackage{algpseudocode}
\usepackage{bm}
\usepackage{multirow}
\usepackage{xcolor}
\usepackage{amsthm}
\newtheorem{theorem}{Theorem}

\defcitealias{stan}{Stan Development Team (2020)}
\newcommand{\pushright}[1]{\ifmeasuring@#1\else\omit\hfill$\displaystyle#1$\fi\ignorespaces}

\algdef{SE}[DOWHILE]{Do}{doWhile}{\algorithmicdo}[1]{\algorithmicwhile\ #1}%
\algnewcommand\algorithmicsave{\quad \textbf{save}}
\algnewcommand\Save{\item[\algorithmicsave]}

\DeclareMathOperator*{\argmax}{arg\,max}
\DeclareMathOperator*{\argmin}{arg\,min}

\begin{document}

\def\spacingset#1{\renewcommand{\baselinestretch}%
{#1}\small\normalsize} \spacingset{1}


 \vspace{-1cm}
 
  \title{\bf Fitting latent non-Gaussian models using variational Bayes and Laplace approximations}
  \author{Rafael Cabral\footnote{Correspondence: rafael.medeiroscabral@kaust.edu.sa} , David Bolin, H\aa vard Rue \hspace{.2cm}\\
    King Abdullah University of Science and Technology, Thuwal, Saudi Arabia
    }
  \date{}
  \maketitle

\begin{abstract}

Latent Gaussian models (LGMs) are perhaps the most commonly used class of models in statistical applications. Nevertheless, in areas ranging from longitudinal studies in biostatistics
to geostatistics, it is easy to find datasets that contain inherently non-Gaussian features, such as sudden jumps or spikes, that adversely affect the inferences and predictions made from an LGM. These datasets require more general latent non-Gaussian models (LnGMs) that can handle these non-Gaussian features automatically. However, fast implementation and easy-to-use software are lacking, which prevent LnGMs from becoming widely applicable. In this paper, we derive variational Bayes algorithms for fast and scalable inference of LnGMs. The approximation leads to an LGM that downweights extreme events in the latent process, reducing their impact and leading to more robust inferences. It can be applied to a wide range of models, such as autoregressive processes for time series, simultaneous autoregressive models for areal data, and spatial Matérn models. To facilitate Bayesian inference, we introduce the \verb|ngvb| package, where LGMs implemented in R-INLA can be easily extended to LnGMs by adding a single line of code.    







%
\end{abstract}

\noindent%

\section{Introduction}  \label{sect:introduction}

Latent models are at the heart of modern statistical modeling and are needed whenever the process of interest is observable only through indirect observations. There exist good statistical methods and a well-established theory for latent Gaussian models (LGMs), where the latent process is assumed to be Gaussian. This class contains generalized linear mixed models \citep{fong2010bayesian}, spatial and spatio-temporal models \citep{bakka2018spatial} and survival models \citep{martino2011approximate}, amongst many other applications. MCMC (Markov Chain Monte Carlo) provides a general recipe to generate samples from LGM posteriors. Still, in many circumstances, especially in high-dimensional problems, the \emph{integrated nested Laplace approximation} (INLA) methodology is faster and more accurate at performing inference within a reasonable amount of time \citep{rue2017bayesian}. As a result, the related R-package (R-INLA; see \url{www.r-inla.org}) has gained much attention and is now a helpful tool for quick and accurate Bayesian inference.


An LGM is a 3-stage hierarchical model, which includes: first, $\pi(\mathbf{y}|\mathbf{x},\boldsymbol{\theta}_1)$ which is a model for the response, where the observations are often conditionally independent given the latent field $\mathbf{x}$ and a set of hyperparameters $\boldsymbol{\theta}_1$; a latent Gaussian field $\mathbf{x}$ with a given mean $\mathbf{m}$ and precision matrix $\mathbf{Q}$ defined conditionally on the second set of hyperparameters $\boldsymbol{\theta}_2$; and finally, a distribution for the priors $\pi(\boldsymbol{\theta})=\pi(\boldsymbol{\theta}_1,\boldsymbol{\theta}_2)$. We will make use of the abbreviation $\mathrm{LGM} \{ \pi(\mathbf{y}|\mathbf{x}, \boldsymbol{\theta}_1), \mathbf{m}(\boldsymbol{\theta}_2), \mathbf{Q}(\boldsymbol{\theta}_2), \pi(\boldsymbol{\theta}) \}$  to refer to these models and $\mathrm{pLGM}\{...\}$ to the ensuing posterior distribution $\pi(\mathbf{x},\boldsymbol{\theta}|\mathbf{y})$.

LGMs allow for non-Gaussian responses, but the latent layer, which often contains Gaussian processes to model spatial and temporal dependence, must be normally distributed. Despite the Gaussian processes' flexible nature, they can over-smooth in the presence of local spikes and sudden jumps in the data. For these cases, a non-Gaussian model often leads to improved predictive power \citep{paciorek2003nonstationary,bolin2014spatial,walder2020bayesian,wallin2015geostatistical}. We refer to these events (local spikes and sudden jumps) as process outliers since they occur with very low probability in Gaussian processes.


%


One way of accounting for outliers and obtaining a more robust analysis is to consider leptokurtic distributions \citep{huber2009}. \cite{west1984outlier} examined outliers in linear regression models using heavy-tailed error distributions and showed that these distributions provide an automatic means of both detecting and accommodating possibly aberrant observations. More recently, \cite{bolin2014spatial} and \cite{wallin2015geostatistical} presented a class of non-Gaussian continuous processes with the same mean and covariance structure as Gaussian processes while allowing for asymmetry and longer-tailed marginal distributions. The main idea consisted in replacing the excitation noise of these processes, which is traditionally Gaussian, with the generalized hyperbolic (GH) distribution \citep{barndorff1978hyperbolic}, which contains the normal, $t$-Student, normal-inverse Gaussian (NIG), and other common distributions as special cases. The same procedure can also be applied to processes defined in discrete space, and \cite{cabral2022controlling} presented a general approach to extend Gaussian processes to non-Gaussianity, which we will review in section \ref{sect:ngproc}.


When we replace the latent Gaussian layer of LGMs with these more robust longer-tailed models, we obtain a class of models that we will call latent non-Gaussian models (LnGMs), and for such models INLA is not applicable. \cite{walder2020bayesian} implemented a Gibbs sampler for some conjugate LnGMs considering Laplace driving noise. A more general implementation in Stan \citepalias{stan}, which considers NIG driving noise, is provided in \cite{cabral2022controlling}. Still, as with LGMs, both of these MCMC implementations are slow when the latent field is high-dimensional.  In this paper, we derive a variational Inference (VI) algorithm for a fast and scalable estimation of LnGMs. Variational inference \citep{bishop2006pattern} is an optimization-based technique for approximate Bayesian inference and provides a computationally efficient alternative to sampling methods (see \cite{zhang2018advances} for recent advances). In our context, the approach consists of fitting an LGM, detecting the outliers in the latent process, then fitting another LGM that downweights the outliers, detecting the outliers again, and repeating this process until convergence is met. Our algorithm often converges in about 3 to 10 iterations, and if INLA \citep{rue2007approximate} is used to fit the LGMs, it is considerably faster than the MCMC alternatives. The implementation of LnGM models is made practical from a user perspective through the package \verb|ngvb|, which we present in section \ref{sect:implementation}.

The distributions we obtain are approximations of the true posterior distributions. However, \cite{wang2019variational} found that the model misspecification error dominates the variational approximation error in the infinite data limit for models where the dimension of the latent variables does not grow with the data. The results suggest that, when it comes to predictive performance, the VI approximation of a well-specified model should be preferred to the exact results from an ill-specified model. Moreover, the example and simulations of sections \ref{sect:illustration} and \ref{sect:simu} demonstrate that, even with few observations, we can approximate the posterior distribution of the latent field $\mathbf{x}$ reasonably well.

 \subsection{Structure of the paper} \label{sect:simplified}

The structure of the paper is as follows. In section \ref{sect:preliminaries} we introduce preliminary concepts necessary to derive this paper's main result. Section \ref{sect:VIresult} contains two theorems that provide the basis for our algorithms, which are shown in section \ref{sect:implementation} to approximate LnGMs' posterior distributions. In section \ref{sect:implementation} we explicate the method with a simple example, and section \ref{sect:simu} contains a simulation study that examines the quality of the approximations. In section \ref{sect:applications} we consider two applications: we fit growth curves using a model with random slopes and intercepts, and we fit areal data with a simultaneous autoregressive model. Finally, \mbox{section \ref{sect:discussion}} contains a discussion and plans for future work.


To simplify the exposition of the results, in sections \ref{sect:preliminaries} and \ref{sect:VIresult}, we restrict the latent layer to have mean $\mathbf{0}$, which will be comprised of one random effect only, for instance, an autoregressive process of order 1 (AR1) to account for temporal dependence of the observations. Also, we will only consider the symmetric NIG distribution as a longer-tailed alternative to the Gaussian distribution. However, the results easily extend to more realistic models with several latent effects and driven by other longer-tailed distributions such as the $t$-Student distribution. These extensions are discussed in appendix \ref{sect:extensions}.



\section{Preliminaries} \label{sect:preliminaries}

In this section, we review some important concepts that will be needed in later sections. Specifically, in section \ref{sect:ngproc}, we summarize a generic class of non-Gaussian models. This allows us to give a precise definition of the class of LnGMs that we are considering in section \ref{sect:LnGM}. Finally, in section \ref{sect:cavi} we introduce the well-known coordinate ascent variational inference (CAVI) algorithm.


\subsection{Extending Gaussian models to non-Gaussianity} \label{sect:ngproc}

To explain the procedure, let us first consider an autoregressive process of order 1 (AR1). The model is defined by the set of equations $\{x_{i}-\rho x_{i-1} = Z_i\}_{i=2,\dotsc,N}$, where the driving noise $Z_i$ usually follows a Gaussian distribution, and $\rho$ is the autocorrelation parameter ($|\rho|<1$). One way to extend this model is to consider a normal inverse-Gaussian (NIG) distribution for the driving noise, which is semi-heavy-tailed, and contains the Gaussian distribution as a special case. This extension preserves the mean and covariance structure of the process while at the same time allowing for more flexible sample paths which exhibit sudden jumps (see Fig. \ref{fig:sim1}) and leptokurtic marginal distributions. Several models can be extended in a similar manner, where linear combinations of elements of $\mathbf{x}$ which are assumed to follow Gaussian noise $[\mathbf{D}\mathbf{x}]_i=Z_i$ will now be driven by NIG noise $\Lambda_i$. In the AR1 example, the dependency matrix that specifies the process is defined by $[\mathbf{D}\mathbf{x}]_i= x_{i}-\rho x_{i-1}$, and it is shown in \mbox{appendix \ref{app:ng}}, along with the dependency matrices for several other models. 


\begin{figure}[htp]
\centering
\includegraphics[width=0.49\linewidth]{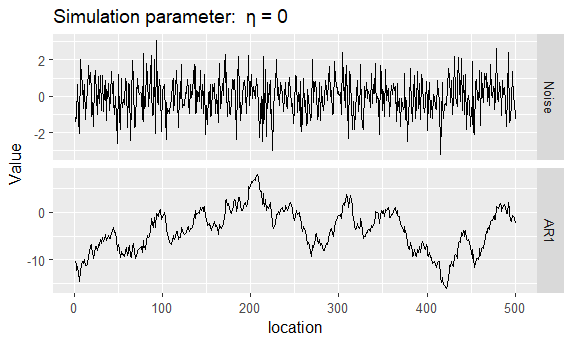}
\includegraphics[width=0.49\linewidth]{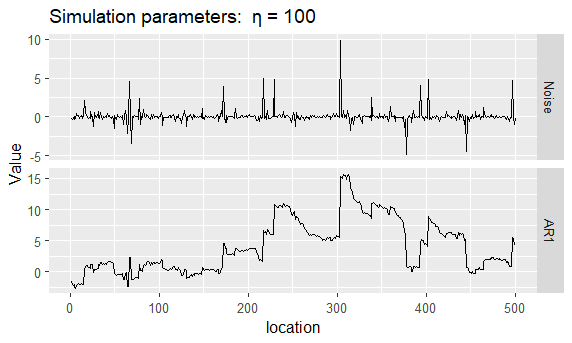}
\caption{Sample of the driving noise (top) and paths (bottom) of an AR1 process. The plots on the left are Gaussian processes, and on the right are NIG-driven processes.}
\label{fig:sim1}
\end{figure}

More generally, let $\mathbf{x}^G$ follow a multivariate Gaussian distribution with dimension $N$, mean $\boldsymbol{0}$, and precision matrix $\mathbf{Q}=  \mathbf{\mathbf{D}}^\top \mathbf{D}$, for a prespecified matrix $\mathbf{D}$. It can be expressed through \begin{equation}\label{eq:gaussian}
\mathbf{D}\mathbf{x}^G\overset{d}{=} \mathbf{Z},
\end{equation}
where $\mathbf{Z} = [Z_1, \dotsc, Z_N]^\top$ is a vector of i.i.d. standard Gaussian variables. The non-Gaussian extension for $\mathbf{x}^G$ consists in replacing the driving noise distribution:
\begin{equation}\label{eq:frame}
\mathbf{D}\mathbf{x}\overset{d}{=} \mathbf{\Lambda},
\end{equation}
where $\boldsymbol{\Lambda} = [\Lambda_1, \dotsc, \Lambda_N]^\top$ is a vector of independent and standardized generalized hyperbolic (GH) random variables that depend on the parameter $\eta$, which controls the non-Gaussianity. For now, we restrict $\Lambda_i$ to be a symmetric NIG distribution, and in appendix \ref{sect:extensions} we consider other member distributions of the GH family. 

\cite{cabral2022controlling} presented these models as a flexible extension of Gaussian models since they contain the Gaussian model as a special case (when $\eta=0$) and deviations from the Gaussian model are quantified by the parameter $\eta$. As $\eta$ increases, the kurtosis of the noises $\Lambda_i$ and of the marginals of $\mathbf{x}$ increase, and in the limiting case $\eta \to \infty$, the NIG distribution converges to the Cauchy distribution. The NIG distribution has a variance-mean mixture representation $\Lambda_i|V_i \sim \mathrm{N}(0, V_i)$ where the mixing variables $V_i$ follows independently an inverse-Gaussian distribution $\mathrm{IG}(1,\eta^{-1})$. Considering the mixing vector $\mathbf{V}=[V_1,\dotsc,V_N]^\top$, the mixture representation for $\mathbf{x}$ is then
\begin{equation} \label{eq:ngaussian}
\mathbf{x}|\mathbf{V} \sim \mathrm{N}\left(\mathbf{0}, \ \  \mathbf{D}^{-1}\mathrm{diag}(\boldsymbol{V})\mathbf{D}^{-T}\right), \ \ \  V_{i}|\eta \overset{ind.}{\sim} \mathrm{IG}(h_i,\eta^{-1} h_i^2),    
\end{equation}
where $h_i$ are predefined constants \citep{cabral2022controlling}, which are equal to 1 for models defined in discrete space.

We list the main models where this extension is possible: i.i.d.~random effects,  random walk (RW) and autoregressive (AR) processes  \citep{ghasami2020autoregressive} for time series; simultaneous autoregressive \citep{walder2020bayesian} and conditional autoregressive processes (CAR) for graphical models and areal data; and Matérn processes \citep{bolin2014spatial, wallin2015geostatistical} which can be used in a variety of applications, such as in geostatistics and spatial point processes. When the smoothness parameter of Matérn models is $\alpha=2$, the sample paths will exhibit spikes, as shown in Fig.~\ref{fig:sim3}. See also \url{rafaelcabral96.github.io/nigstan/} for Bayesian applications of these models.


\begin{figure}[htp]
\centering
\includegraphics[width=\linewidth]{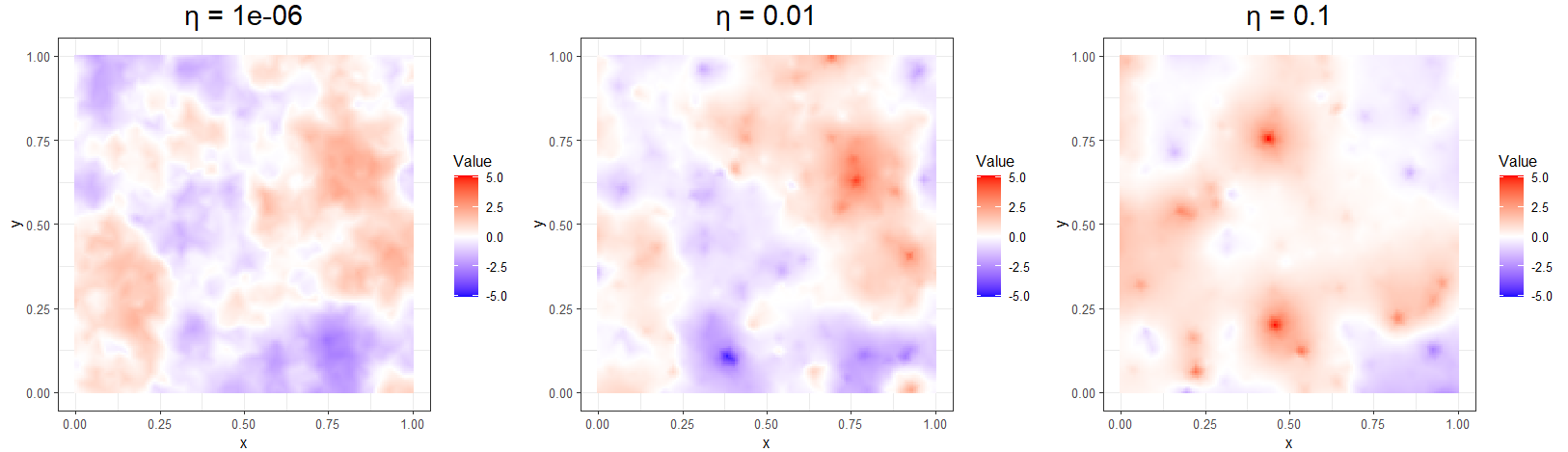}
\caption{Sample paths of a non-Gaussian Matérn model in 2D (smoothness parameter $\alpha=2$), driven by NIG noise for several values of $\eta$.}
\label{fig:sim3}
\end{figure}


\subsection{Latent non-Gaussian models} \label{sect:LnGM}


We construct LnGMs by replacing the multivariate Gaussian assumption on the latent field $\mathbf{x}$ of LGMs with the non-Gaussian model in  \eqref{eq:ngaussian}. The hierarchical structure is as follows: 
\begin{equation}\label{eq:lngm}
    \begin{array}{cc}
\mathrm{Response} & \mathbf{y} \mid \mathbf{x}, \boldsymbol{\theta}_{1} \sim \prod_{i \in \mathcal{I}} \pi\left(y_{i} \mid x_{i}, \boldsymbol{\theta}_{1}\right) \\ \\
\mathrm{Latent \ field} & \mathbf{x}|\boldsymbol{V},\boldsymbol{\theta}_{2} \sim \mathrm{N}\left(\mathbf{0}, \ \  \mathbf{D}(\boldsymbol{\theta}_{2})^{-1}\mathrm{diag}(\boldsymbol{V})\mathbf{D}(\boldsymbol{\theta}_{2})^{-T}\right) \\  \\
\mathrm{Mixing \ variables} & V_{i}|\eta \overset{ind.}{\sim} \mathrm{IG}(h_i,\eta^{-1} h_i^2)   \\ \\ 
\mathrm{Hyperparameters} & \boldsymbol{\theta} \sim \pi(\boldsymbol{\theta}), \ \ \eta \sim \mathrm{Exp}(\alpha_\eta)
\end{array}
\end{equation}

The main difference with the LGM structure presented in section \ref{sect:introduction} is that we further condition the latent field $\mathbf{x}$ on a vector of mixing variables $\mathbf{V}$, which enter the covariance matrix, adding more flexibility to the model. In addition, we select an exponential prior for $\eta$, which prevents overfitting the data as explained in  \cite{cabral2022controlling}. We could see the mixing variables as hyperparameters which would turn  \eqref{eq:lngm} into an LGM that we could fit in INLA. However, to guarantee fast computations and good accuracy, the INLA methodology requires the number of hyperparameters to be smaller than 20, which is generally not the case since the dimension of $\mathbf{V}$ is the same as the dimension of $\mathbf{x}$.


\subsection{Coordinate ascent variational inference} \label{sect:cavi}

Variational inference (VI) methods \citep{bishop2006pattern, jordan1999introduction,wainwright2008graphical}, instead of relying on traditional Markov Chain Monte Carlo sampling schemes to approximate a posterior distribution $\pi(\mathbf{z}|\mathbf{y})$, finds a surrogate density $q(\mathbf{z})$ that solves the optimization problem:
\begin{equation}\label{eq:viequation}
q(\mathbf{z}) = \argmin_{q \in \mathcal{Q}} \left\{\mathrm{KLD}(q(\mathbf{z})|\pi(\mathbf{z}|\mathbf{y}))\right\} = \argmax_{q \in \mathcal{Q}} E_{q(\mathbf{z})}\left( \log\left(\frac{\pi(\mathbf{y},\mathbf{z})}{q(\mathbf{z})}\right)\right),    
\end{equation}
where KLD stands for the Kullback-Leibler divergence. The expectation on the right of  \eqref{eq:viequation} is the evidence lower bound (ELBO), and maximizing the ELBO is the same as minimizing the KLD since the relationship between them is given by
$$
\mathrm{ELBO}(q(\mathbf{z})) = -\mathrm{KLD}(q(\mathbf{z})|\pi(\mathbf{z}|\mathbf{y})) +\log \pi(\mathbf{y}),
$$
where $\pi(\mathbf{y})$ is the evidence. Since the KLD is always non-negative, the ELBO provides a lower bound on the log-evidence.  \cite{wang2019frequentist}, and references therein provide theoretical results on VI methods, namely, the authors extend the theory of Bernstein-von Mises to the variational posterior and establish frequentist consistency and asymptotic normality of VI methods.




 The space of searched functions $q(\mathbf{z})$ is restricted to a family of functions $\mathcal{Q}$, which should be flexible enough to allow for a reasonable approximation but simple enough for efficient optimization. The mean-field variational family assumes independence between each unknown variable in $\mathbf{z}$: $q(\mathbf{z})=\prod_{i=1}^N q(z_i)$. If we use $q(\mathbf{z})=\prod_{i=1}^N q(z_i)$ as a surrogate density, then the solution of the variational problem in  \eqref{eq:viequation} satisfies the system:
\begin{equation} \label{eq:CAVI}
q_i\left(z_i\right) \propto \exp \left\{E_{-i}\left(\log \pi\left(\mathbf{z}, \mathbf{y}\right)\right)\right\}, \ \ \mathrm{for} \ i=1,\dotsc,N.    
\end{equation}

Tutorial style introductions on the mean-field variational inference approach can be found in \cite{blei2017variational} and \cite{tran2021practical}. One of the most popular algorithms to maximize the ELBO is the coordinate ascend variational inference (CAVI), which iteratively updates each factor $q(z_i)$ of  \eqref{eq:CAVI} until the ELBO reaches a local optimum \citep{bishop2006pattern,blei2017variational}. A disadvantage of the mean-field family is that it cannot capture correlations between the unknown parameters. It is possible, however, to partition $\mathbf{z}$ into $k$ blocks $\mathbf{z}^{(1)}, \mathbf{z}^{(2)}, \dotsc, \mathbf{z}^{(k)}$, and assume independence between the elements of different blocks but allow for dependencies between the elements of each block: ${q(\mathbf{z})=\prod_{i=1}^k q(\mathbf{z}^{(i)})}$. This is called structured variational inference \citep{saul1995exploiting, barber1998tractable,zhang2018advances}, which we will make use of in sections \ref{sect:CAVIfull} and \ref{sect:CAVIcollapse} to account for spatial or temporal dependencies in the posterior inferences. 

\section{Variational inference for LnGMs} \label{sect:VIresult}

We present in this section the main theoretical results that allow us to construct the CAVI algorithm to approximate the posterior distribution of $\mathbf{z} = (\mathbf{x}, \boldsymbol{\theta}, \mathbf{V}, \eta)$. 


 
\subsection{Structured VI} \label{sect:CAVIfull}

\sloppy In the structural VI approach, we search for the optimal surrogate density ${q(\mathbf{x}, \boldsymbol{\theta}, \mathbf{V},\eta) = q(\mathbf{x}, \boldsymbol{\theta}) q(\mathbf{V})q(\eta)}$, where the only restriction in the space $\mathcal{Q}$ we are imposing is the posterior independence between the blocks $(\mathbf{x}, \boldsymbol{\theta})$, $\mathbf{V}$ and $\eta$. The result is given in Theorem \ref{theo:1}, where $q(\mathbf{x}, \boldsymbol{\theta})$ is the posterior distribution of an LGM, and the mixing variables $V_i$ and parameter $\eta$ follow a generalized inverse Gaussian distribution (GIG). The GIG distribution has pdf:
\begin{equation} \label{eq:gig}
    \pi_{\mathrm{GIG}}(x; p, a, b) =\frac{(a / b)^{p / 2}}{2 K_p(\sqrt{a b})} x^{(p-1)} e^{-(a x+b / x) / 2}, \quad x>0,
\end{equation}
where $K_\lambda(x)$ is the modified Bessel function of the second kind of order $\lambda$. 

\begin{theorem}\label{theo:1} 
\sloppy The surrogate density $q(\mathbf{x}, \boldsymbol{\theta}, \mathbf{V},\eta) = q(\mathbf{x}, \boldsymbol{\theta}) q(\mathbf{V})q(\eta)$ that minimises  $\mathrm{KLD}(q(\mathbf{x}, \boldsymbol{\theta}, \mathbf{V}, \eta)|\pi(\mathbf{x}, \boldsymbol{\theta}, \mathbf{V}, \eta|\mathbf{y}))$ is a solution of the system:
\begin{align*}
q(\mathbf{x}, \boldsymbol{\theta}) &\sim  \mathrm{pLGM}  \{ \pi(\mathbf{y}|\mathbf{x},\boldsymbol{\theta}_1), \  \mathbf{m}=\mathbf{0}, \  \mathbf{Q}  = \mathbf{D}(\boldsymbol{\theta}_2)^\top \mathrm{diag} (\mathbf{V}^{(-)}) \mathbf{D}(\boldsymbol{\theta}_2), \  \pi(\boldsymbol{\theta}) \}, \\ 
 q(V_i) &\sim \mathrm{GIG}\left( -1, \ E_{q(\eta)}(\eta^{-1}), \ E_{q(\mathbf{x}, \theta)}([\mathbf{D}\mathbf{x}]_i^2]) + h_i^2E_{q(\eta)}(\eta^{-1}) \right), \ \ i=1,\dotsc,N, \\
 q(\eta) &\sim \mathrm{GIG}\left( -N/2 + 1, \ 2\alpha_\eta, \ \sum_{i=1}^N E_{q(V_i)}(V_i) -2h_i + h_i^2 E_{q(V_i)}(V_i^{-1})  \right), 
\end{align*}
where $V_i^{(-)} = E_{q(V_i)}(V_i^{-1})$.

\end{theorem}
\begin{proof}
See appendix \ref{app:proof1}.
\end{proof}

%

\subsection{Structured and collapsed VI} \label{sect:CAVIcollapse}
Collapsed variational inference (CVI) relies on the idea of analytically integrating certain model parameters (see \cite{zhang2018advances} and references therein). Due to the reduced number of parameters to be estimated and the removal of hierarchical correlations, the inference is typically faster. For the LnGM model in  \eqref{eq:lngm} we can integrate out $\eta$ from $\pi(\mathbf{V}|\eta)$ and obtain $\pi(\mathbf{V})$:
\begin{equation} \label{eq:distV}
    \pi(\mathbf{V}) = \int_{0}^{\infty} \left(\prod_{i=1}^N \pi(V_i|\eta)\right) \pi(\eta) d\eta, 
\end{equation}
We present in Theorem \ref{theo:2} the collapsed version of Theorem \ref{theo:1}.



\begin{theorem} \label{theo:2} 
\sloppy The surrogate density $q(\mathbf{x}, \boldsymbol{\theta}, \mathbf{V}) = q(\mathbf{x}, \boldsymbol{\theta}) q(\mathbf{V})$ that minimises $\mathrm{KLD}(q(\mathbf{x}, \boldsymbol{\theta}, \mathbf{V})|\pi(\mathbf{x}, \boldsymbol{\theta}, \mathbf{V}|\mathbf{y}))$ satisfies the system:
\begin{align*}
    q(\mathbf{x}, \boldsymbol{\theta}) &\sim \mathrm{pLGM} \{ \pi(\mathbf{y|\mathbf{x},\boldsymbol{\theta}_1}), \  \mathbf{m}=\mathbf{0}, \  \mathbf{Q}  = \mathbf{D}(\boldsymbol{\theta}_2)^\top \mathrm{diag} (\mathbf{V}^{(-)}) \mathbf{D}(\boldsymbol{\theta}_2), \  \pi(\boldsymbol{\theta}) \} \\ 
    q(\mathbf{V}) &\sim \int_{0}^{\infty} \left(\prod_{i=1}^N \pi_{\mathrm{GIG}}(V_i; -1,  \eta^{-1}, h_i^2\eta^{-1} + E_{q(\mathbf{x}, \boldsymbol{\theta}_2)}([\mathbf{D}\mathbf{x}]_i^2) ) \right) q(\eta )d\eta,
\end{align*} 
where $V_i^{(-)} = E_{q(V_i)}(V_i^{-1})$ and 
$q(\eta) \propto \eta^{-N/2}e^{\eta^{-1}(\sum_{i=1}^N h_i) - \alpha_\eta \eta} \prod_{i=1}^N\left( \frac{ K_{-1}(\sqrt{\eta^{-1}(d_i+h_i^2\eta^{-1})}) }{\sqrt{d_i\eta + h_i^2}}\right)$.
\end{theorem}
\begin{proof}
See appendix \ref{app:proof2}.
\end{proof}
\sloppy Since $q(\mathbf{V})$ in Theorem \ref{theo:2} has the form $\int_0^\infty (\prod_{i=1}^N q(V_i|\eta)) q(\eta) d\eta$ then samples from $q(\mathbf{V})$ can be obtained by first sampling from $q(\eta)$ and then sampling from ${q(V_i | \eta) \sim \mathrm{GIG}(-1, \eta^{-1}, h_i^2\eta^{-1} +  E_{q(\mathbf{x}, \boldsymbol{\theta}_2)}([\mathbf{D}\mathbf{x}]_i^2))}$ for each element $V_i$. Also, since $q(\eta)$ is one-dimensional, an inverse transform sampler can be built by numerically approximating the quantile function.

\subsection{Extensions}

As mentioned in section \ref{sect:simplified}, the results can be extended for models with several latent effects, for instance, one latent effect for time and another for space. Moreover, we can consider other longer-tailed distributions for the driving noise, such as the $t$-Student distribution. We discuss these extensions in appendix \ref{sect:extensions}.

\section{Implementation and R package} \label{sect:implementation}

We present in this section two CAVI algorithms that perform approximate inference of LnGMs based on Theorems \ref{theo:1} and \ref{theo:2}. These algorithms are implemented in the R package \verb|ngvb|, which we present at the end of this section. 
	
\subsection{CAVI algorithms}

Algorithm \ref{alg:MFVI1} is the CAVI algorithm that iteratively update the surrogate densities  $q(\mathbf{x}, \boldsymbol{\theta})$, $q(V_i)$ and $q(\eta)$ of the structured variational inference (SVI) approximation. The SVI algorithm requires computing moments of order $k$ of the GIG distribution in  \eqref{eq:gig} which we abbreviate to: $\mathrm{m.GIG(order=k, p, a, b)}$. We do not use the ELBO as a criterion for convergence because it is too expensive to compute, and instead, we stop the algorithm when the change of the moment $E_{q(\eta)}[\eta]$ has fallen below a given threshold. The structured and collapsed variational inference (SCVI) algorithm is shown in Algorithm \ref{alg:MFVI2}.


\begin{algorithm}[h]
\caption{Structured VI (SVI) algorithm for LnGMs}\label{alg:MFVI1}
\begin{algorithmic}[1]
\Require  INLA,  $\mathbf{D(\boldsymbol{\theta}_2)}$, $\mathbf{h}$
\State $N \gets$ dimension of latent field $\mathbf{x}$;
\State $\mathbf{V}^{(-)} \gets \mathbf{h}$;
\State $\eta^{(-)} \gets 0.5$;

\Do
	\State $q(\mathbf{x},\boldsymbol{\theta)} \sim \mathrm{pLGM} \{ \pi(\mathbf{y|\mathbf{x},\boldsymbol{\theta}_1}),   \mathbf{m}=\mathbf{0},   \mathbf{Q}  = \mathbf{D(\boldsymbol{\theta}_2)}^\top \mathrm{diag} (\mathbf{V}^{(-)}) \mathbf{D(\boldsymbol{\theta}_2)}, \pi(\boldsymbol{\theta}) \}$
 
    \For{$i \in \{1,\dotsc,N\}$}
        \State $d_i \gets E_{q(\mathbf{x},\boldsymbol{\theta}_2)}([\mathbf{D}(\boldsymbol{\theta}_2)\mathbf{x}]_i^2)$; 
        \State $V_{i}^{(-)} \gets$ m.GIG(order = -1, p = -1, a = $\eta^{(-)}$, b = $d_i + h_i^2\eta^{(-)}$); 
        \State $V_{i}^{(+)} \gets$ m.GIG(order = 1, p = -1, a = $\eta^{(-)}$, b = $d_i + h_i^2\eta^{(-)}$);

    \EndFor
    \State $\eta^{(-)} \gets $ m.GIG(order = -1, p = $-N/2 + 1$, a = $2\alpha_\eta$, b = $\sum_{k=1}^N V_{k}^{(+)} - 2h_i + h_i^2V_{k}^{(-)}$); 

\doWhile{$E[q(\eta)]$ does not converge}
\end{algorithmic}
\end{algorithm}

\begin{algorithm}[h]
\caption{Structured and collapsed VI (SCVI) algorithm for LnGMs}\label{alg:MFVI2}
\begin{algorithmic}[1]
\Require  INLA, $\mathbf{D(\boldsymbol{\theta}_2)}$, $\mathbf{h}$, $m$
	\State $N \gets$ dimension of latent field $\mathbf{x}$;
    \State $\mathbf{V}^{(-)} \gets \mathbf{h}$;
    \Do
		\State Steps 5-7 of Algorithm \ref{alg:MFVI1}
    	\State Get $m$ samples from $q(\mathbf{V})$; \Comment{Using Theorem \ref{theo:2} }
    	\State $\mathbf{V}^{(-)} \gets E_{q(\mathbf{V})}( [V_1^{-1}, V_2^{-1}, \dotsc, V_N^{-1}]^\top)$;  \Comment{Approximated by Monte Carlo}

\doWhile{$E[q(\eta)]$ does not converge}  
\end{algorithmic}
\end{algorithm}

\subsection{Fitting LGMs using INLA}

The algorithms presented before are essentially a wrapper to algorithms that fit LGMs since they recursively fit LGMs for fixed values of the mixing variables $\mathbf{V}$ until convergence. The speed of both CAVI algorithms is then largely determined by the time it takes to fit each LGM and how many iterations are needed to obtain convergence.

Under conjugacy, the posterior distributions of the parameters of LGMs can be found analytically. Generally, however, we have to resort to approximation methods, such as MCMC or Integrated Nested Laplace Approximations (INLA), to fit LGMs \citep{rue2007approximate, rue2017bayesian}. For LGMs,  MCMC sampling can be hampered by slow computation time and poor convergence, and INLA can provide accurate and fast analytical approximations \citep{opitz2017latent, blangiardo2013spatial}. 

First, INLA approximates the posterior of the hyperparameters $\boldsymbol{\theta}$ by computing:
$$
\left.\tilde{\pi}_{\mathrm{LA}}(\boldsymbol{\theta} \mid \mathbf{y}) \propto \frac{\pi\left(\mathbf{x}, \boldsymbol{\theta} \mid \mathbf{y}\right)}{\tilde{\pi}_{\mathrm{G}}\left(\mathbf{x} \mid \boldsymbol{\theta}, \mathbf{y}\right)}\right|_{\mathbf{x}=x^*(\boldsymbol{\theta})}
$$
where $\tilde{\pi}_{\mathrm{G}}\left(\mathbf{x}^* \mid \boldsymbol{\theta}, \mathbf{y} \right)$ is the Gaussian approximation obtained by matching the mode ($x^*(\boldsymbol{\theta})$) and curvature at the mode of the full joint density $\pi(\mathbf{x} \mid \boldsymbol{\theta}, \mathbf{y})$. The next step is to approximate the joint posterior, which will be a mixture of skew-Gaussian copula densities $\pi_{\mathrm{SGC}}$ \citep{chiuchiolo2022extended}, given by:
$$
\pi(\boldsymbol{x}, \boldsymbol{\theta} \mid \boldsymbol{y}) \propto \sum_k \pi_{\mathrm{SGC}}(\boldsymbol{x} \mid \boldsymbol{\theta}_k, \boldsymbol{y}) \pi(\boldsymbol{\theta}_k \mid \boldsymbol{y})  \Delta_k.
$$
The previous approximation considers only a set of hyperparameter values $\left\{\boldsymbol{\theta}_k\right\}_{k=1}^K$ with associated integration weights $\left\{\Delta_k\right\}_{k=1}^K$. INLA obtains these integration points by placing a regular grid about the posterior mode of $\boldsymbol{\theta}$ or using a central composite design centered at the posterior mode. We use this approximation to compute step 7 of Algorithm \ref{alg:MFVI1}.

\subsection{The ngvb package}

Algorithms \ref{alg:MFVI1} and \ref{alg:MFVI2} are implemented in the R package \verb|ngvb| (non-Gaussian variational Bayes), which uses INLA to fit the LGM of step 5 of Algorithm \ref{alg:MFVI1}. The development version of this package can be found in \url{github.com/rafaelcabral96/ngvb}. 

\sloppy The most convenient manner of fitting LnGMs with \verb|ngvb| is by first fitting an LGM model using INLA and then using the \verb|inla| object as the input of the \verb|ngvb| function. For instance, fitting a latent non-Gaussian RW1 takes the form:

\begin{verbatim}
 data    <- list(x = 1:100, y = g(1:100))
 LGM     <- inla(formula = y ~  f(x,  model = "rw1"), data = data)
 LnGM    <- ngvb(fit = LGM, selection = list(x=1:100))
\end{verbatim}

The \verb|selection| argument specifies which components of the LGM are to be extended to non-Gaussianity. This functionality is currently available for i.i.d., RW1, RW2, and AR1 models, and we plan to extend it to other commonly used models. Alternatively, other LnGMs can be specified with the argument \verb|manual.configs| as detailed in the package documentation. The package uses, by default, the SCVI approximation. The SVI approximation or a built-in Gibbs sampler can be chosen through the argument \verb|method = "SVI"| or \verb|method = "Gibbs"|. The built-in Gibbs sampler iterates between the full conditionals of $(\mathbf{x},\boldsymbol{\theta})$, $\mathbf{V}$ and $\eta$. For each iteration, it finds the full conditional of $(\mathbf{x},\boldsymbol{\theta})$ by fitting an LGM in INLA, which is time-consuming but can be used in low-dimensional problems. More details about the Gibbs sampler can be found in appendix \ref{sect:Gibbs}.
 
  The R generic functions \verb|print|, \verb|summary|, \verb|plot|, and \verb|fitted| are available to process the output of the \verb|ngvb| function. Also, the function \verb|simulate| can be used to obtain samples from the posterior distributions of $\mathbf{x}, \boldsymbol{\theta}, \mathbf{V}$, and $\eta$.

\subsection{Improved predictions}
The surrogate $q(\mathbf{x})$ is the posterior distribution of an LGM, which downweights the influence of outlier events. Under weak conditions \cite{chiuchiolo2022extended} showed that for LGMs, $\pi(\mathbf{x}|\mathbf{y}) \leq C \pi(\mathbf{x})$, where $\pi(\mathbf{x})$ is the Gaussian prior for $\mathbf{x}$ and $C$ is a constant. Therefore, the marginals of $q(\mathbf{x})$ cannot have tails heavier than Gaussian tails. This happens because, in the VI approximation, the latent field $\mathbf{x}$ and mixing variables $\mathbf{V}$ have independent posterior distributions. We can, however, obtain leptokurtic approximate posterior samples of $\mathbf{x}$. We first sample from $q(\mathbf{V})$ and $q(\boldsymbol{\theta})$ obtained with the VI approximation, and then we generate samples from $\pi(\mathbf{x}|\mathbf{y},\mathbf{V},\boldsymbol{\theta})$ using INLA. In the \verb|ngvb| package, leptokurtic samples of $\mathbf{x}$ can be obtained by adding the argument \verb|improved.tail=TRUE| in the \verb|simulate| method.

\section{Illustration} \label{sect:illustration}

In this illustrative example, we take inspiration from the longitudinal study in \cite{asar2020linear} in which measurements related to the kidney function of several patients were recorded over time. The goal was to predict kidney function from noisy measurements. LGMs did not adapt well to sudden drops in measurements, which was problematic since these drops are an example of ``acute kidney injury", which should prompt an immediate medical intervention. We consider here a simpler simulated example that demonstrates the same phenomenon. The longitudinal model was the following:
$$
y_{ij} =\sigma_x \mathbf{x}_{i} +  \sigma_\epsilon\epsilon_{ij},     \ \ \ \ i=1,\dotsc,10, \ \ j=1,\dotsc,100,
$$
\sloppy where $\epsilon_{ij}$ are independent Gaussian noise to model within-subject error, the vectors $\mathbf{x}_{i}$ are mutually independent, and each follows a Matérn model with fixed range parameter ${\kappa=0.001}$. For the simulated data shown in Fig. \ref{fig:longdata}, we fixed $\sigma_x=1$, $\sigma_\epsilon=0.1$, standardized each replicate, and for the last replicate, we further added sudden jumps of sizes -6 and 11 at time points 20 and 40. Finally, we fitted the simulated data to an LGM and LnGM, where the previous Matérn model was driven by Gaussian and NIG noise, respectively. 

\begin{figure}[h]
\centering
\includegraphics[width=\linewidth]{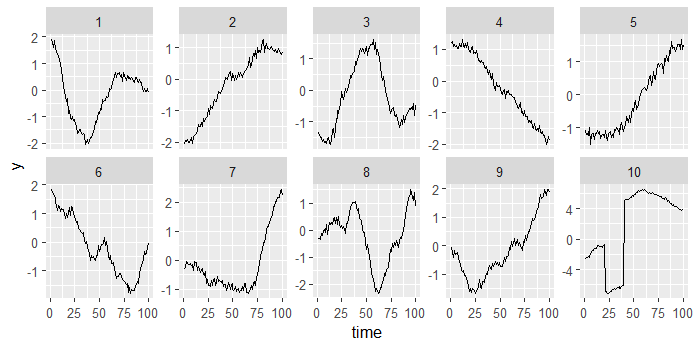}
\caption{Simulated longitudinal data.}
\label{fig:longdata}
\end{figure}

We fitted the LGM to this data using INLA and the LnGM using STAN and the SVI and SCVI algorithms in the $\verb|ngvb|$ package, utilizing the default INLA priors and the prior $\eta \sim \text{Exp}(5)$. We show the smoothed processes for the first and last replicates in Fig. \ref{fig:longfit}.

The price to pay for settling with an LGM that cannot accommodate the two jumps is more uncertainty in the smoothed process and less accurate predictions. For the LGM, the marginal standard deviation of the smoothed process increased from 0.04 to 0.14 when we added the two jumps to the last realization, while for the LnGM, it remained at 0.04. This example shows the non-robustness of the LGM as a result of discrepancies in the data. Two sudden jumps in 1000 observations were sufficient to essentially triple the marginal standard deviation of the smoothed process.

 We can also see that the sudden jumps were oversmoothed for the LGM while the LnGM captured them well. Furthermore, the Bayesian leave-one-out estimate of out-of-sample predictive fit \citep{vehtari2017practical} increased from -341 to 685 when we extended the LGM to an LnGM using SVI and SCVI approximations. Finally, the smoothed processes obtained with the VI approximations were essentially the same as the ones obtained from STAN. However, the STAN fit took 35 minutes (4 parallel chains, 400 warmup, 800 sampling iterations), while the VI algorithms converged in less than a minute.

\begin{figure}[h]
\centering
\includegraphics[width=0.45\linewidth]{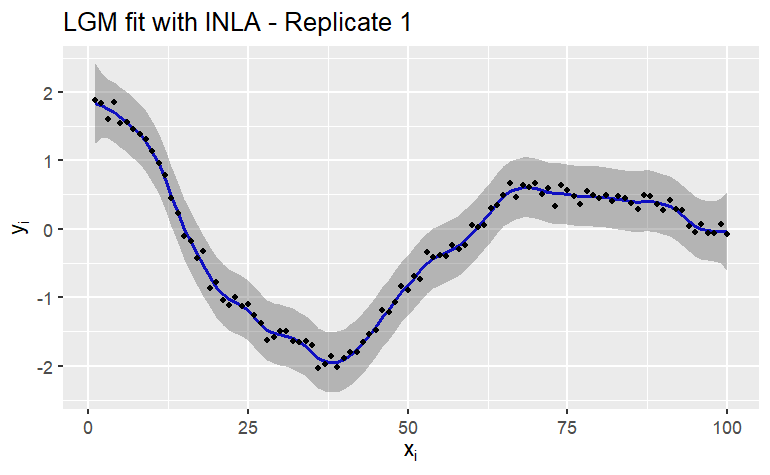}
\includegraphics[width=0.45\linewidth]{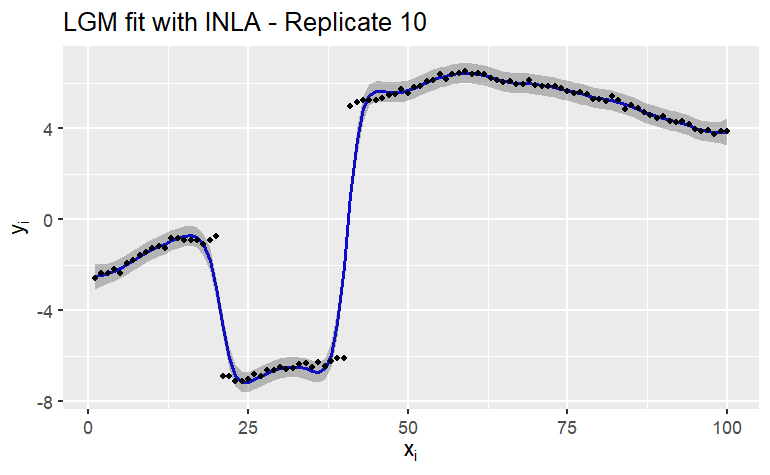}
\includegraphics[width=0.45\linewidth]{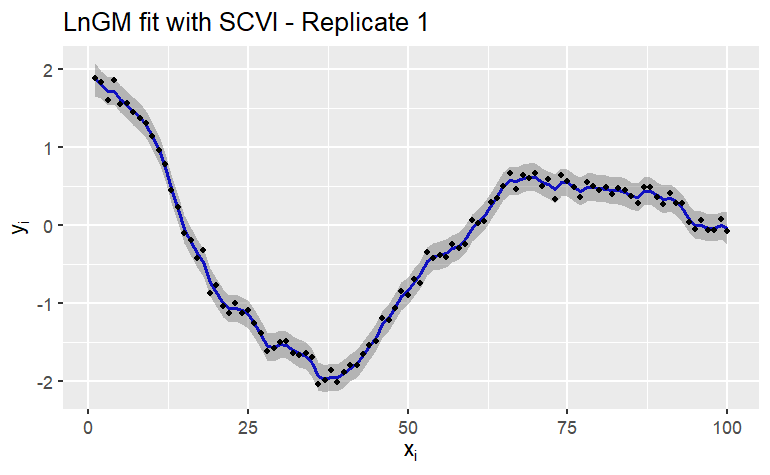}
\includegraphics[width=0.45\linewidth]{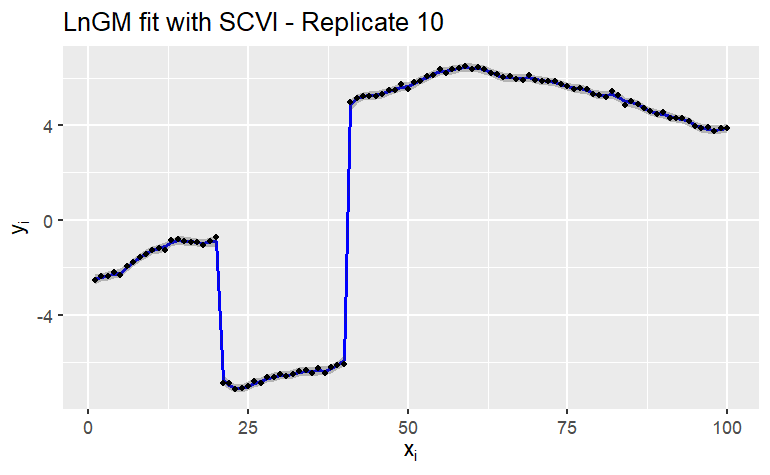}
\includegraphics[width=0.45\linewidth]{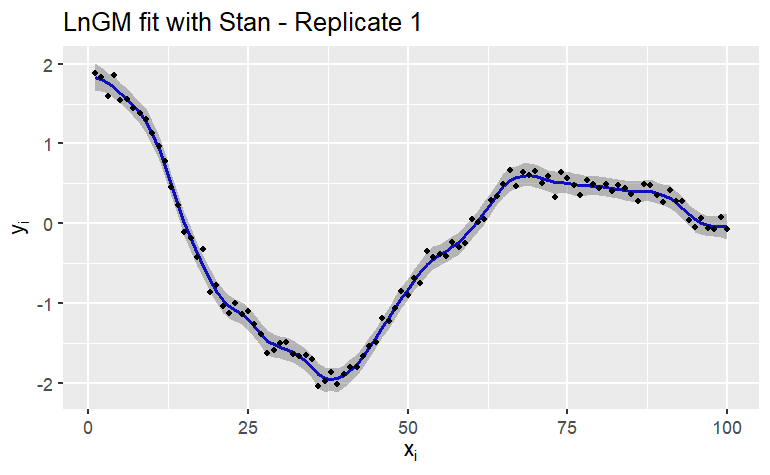}
\includegraphics[width=0.45\linewidth]{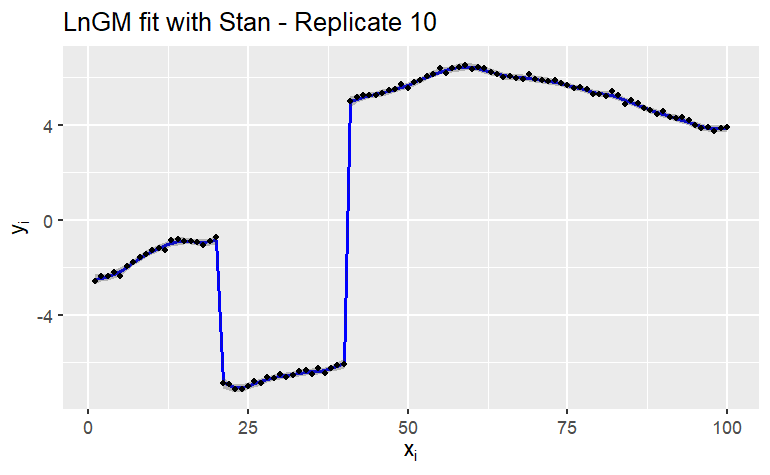}
\caption{The data is shown as black circle markers. The lines represent the posterior mean of the latent process and the shaded area the 97.5\% credible intervals.}
\label{fig:longfit}
\end{figure}

 We show in Fig.~\ref{fig:progression} the progression of the posterior means of $V_{ij}$ for  ${i=1,\dotsc,10,}$ ${j=1,\dotsc,100},$ for both the SVI and SCVI algorithm. The SVI algorithm converged in 25 iterations, while the SCVI algorithm converged in 7 iterations, roughly. We can immediately recognize 4 mixing variables with large posterior means related to the last replicate, which correspond to time points 21, 22, 41, and 42, right after the time points of the sudden jumps. For these mixing variables, the VI approximations $q(V_{ij})$ tend to underestimate the true posterior mean and variance, as will be shown in the simulation study in section \ref{sect:simu}. However, having a good approximation for the latent process $q(\mathbf{x})$ is often the main interest.

\begin{figure}[h]
   \centering
  \includegraphics[width=0.49\linewidth]{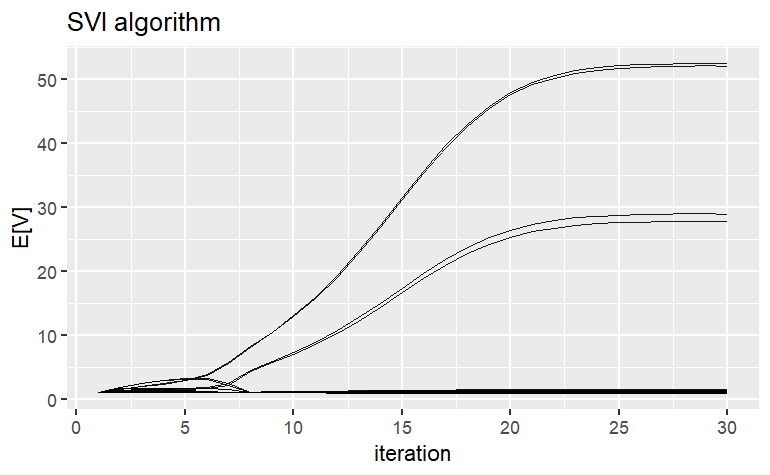} 
  \includegraphics[width=0.49\linewidth]{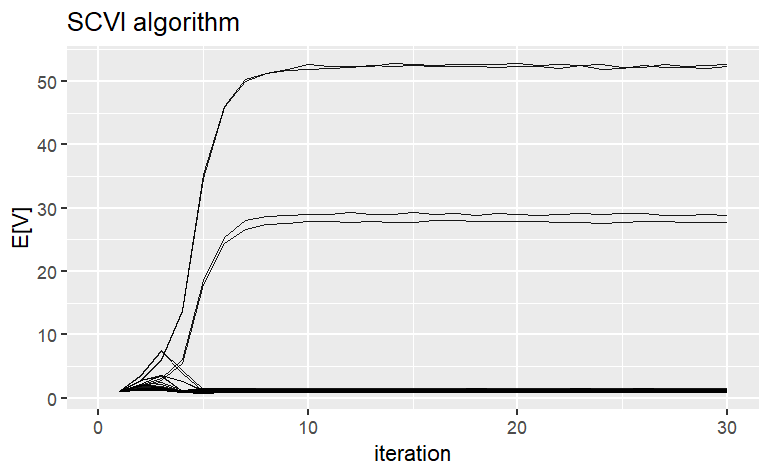} 
  \caption{Means of $q(V_i)$ across the iterations for the SVI algorithm (left) and SCVI algorithm (right). }
  \label{fig:progression}
\end{figure}

\subsection{Interpretation} \label{sect:interpretation}
 Robustness can be obtained by a model that appropriately downweights extreme events \citep{box1980sampling, o2012bayesian, desgagne2015robustness}. \cite{box1980sampling} also said, ``Efficient model building requires both diagnostic checking and model robustification, where
 by robustification I mean judicious and grudging elaboration of the model to ensure protection against particular hazards".   
 
To understand how model robustification happens automatically for LnGMs, we can look at the ``messages'' that are exchanged between the surrogate densities $q(\mathbf{x}, \boldsymbol{\theta})$, $q(\mathbf{V})$ and $q(\eta)$ in Theorem \ref{theo:1}, which are shown in Fig. \ref{fig:messages}. 

\begin{figure}[h]  
  \center
  \includegraphics[width=0.6\linewidth]{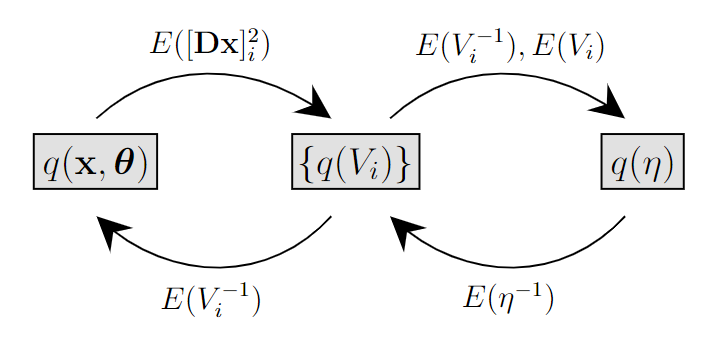}
   \caption{Messages exchanged between the surrogate densities in  Theorem \ref{theo:1}.}
    \label{fig:messages}
\end{figure}

 The surrogate $q(V_i)$ receives one message from $q(\mathbf{x}, \boldsymbol{\theta})$, which is $d_i = E_{q(\mathbf{x}, \theta)}([\mathbf{D}\mathbf{x}]_i^2])$ and another message from $q(\mathbf{\eta})$ which is $E_{q(\eta)}(\eta^{-1})$. The first message conveys how many and how large are the extreme events in the latent field. For the AR1 model, $[\mathbf{D}\mathbf{x}]_i= x_{i}-\rho x_{i-1}$ will be large whenever there is a sudden jump at location $i$ (when the autocorrelation $\rho$ is high). Thus, diagnostic checking is manifested in our algorithm by the identification of particular path features that are unlikely in the Gaussian model, which in this case are sudden jumps. The second message conveys how much non-Gaussianity is present at the current iteration (the smaller this message, the more non-Gaussianity is believed to be present). Then, the distribution of $q(V_i)$ is given by a compromise between these two messages, and samples of $V_i$ will typically take larger values the higher $d_i$ and the smaller $E_{q(\eta)}[\eta^{-1}]$. The surrogate $q(\mathbf{x}, \boldsymbol{\theta})$ is a LGM proportional to  $\pi(\mathbf{y} | \mathbf{x}, \boldsymbol{\theta}_1)\pi(\mathbf{x}|\mathbf{V}^{(-)},\boldsymbol{\theta}_2)\pi(\boldsymbol{\theta})$ and it receives the messages $V_i^{(-)} = E_{q(V_i)}[V_i^{-1}], \ i=1,\dotsc,N$.  Downweigthing of outlier events in the latent process occurs because:
  \[\log \pi(\mathbf{x}|\mathbf{V}^{(-)}, \boldsymbol{\theta}) = -\frac{1}{2} \mathbf{x^\top D^\top}\mathrm{diag}(\mathbf{V}^{(-)}) \mathbf{D x} + \mathrm{const} = -\frac{1}{2} \times \sum_{i=1}^N  [\mathbf{D}\mathbf{x}]_i^2/V_i^{(-)} + \mathrm{const},\] 
  therefore when $V_i^{(-)}$ is large, the influence of a large event of $[\mathbf{D}\mathbf{x}]_i$  is reduced.
 
 

  \section{Simulation study} \label{sect:simu}
  

  If the number of latent variables increases as the data grows, the frequentist consistency of VI methods in \cite{wang2019frequentist} does not apply. We ran a simulation to study the quality of the VI approximations in this case. We considered the model $y_i = \sigma_x x_i + \sigma_y \epsilon_i, \ i=1,\dotsc,N$, where $\epsilon_i$ is standard Gaussian noise, and $\mathbf{x}=[x_1,\dotsc,x_{N}]^\top$ follows a non-Gaussian AR1 prior, defined by $x_{i}=\rho x_{i-1} + \Lambda_i, i=2,\dotsc,N$, where $\rho=0.9$ and $\Lambda_i$ is NIG noise with parameter $\eta$. To fit the models we chose the priors $\tau_y = 1/\sigma_y^2 \sim \mathrm{Gam}(1,0.5)$, $\tau_x = 1/\sigma_x^2 \sim \mathrm{Gam}(1,0.5)$ (using the shape and rate parameterization) and $\eta \sim \mathrm{Exp}(5)$.
  
  We fitted the previous model to simulated data considering several scenarios for the time series. We varied the dimension $N \in \{100,500,1000\}$ and the non-Gaussianity parameter $\eta \in \{0, 0.5, 1, 5, 10, 100\}$. The other parameters were $\sigma_y=1, \sigma_x=2$, and $\rho=0.9$. We ran several replications so that there would be at least 10000 approximations of the mixing variables $V_i$ to analyze in each scenario. We fitted the models in Stan and with $\verb|ngvb|$ using the SVI and SCVI approximations. Finally, we compared the posterior means and standard deviations of $\mathbf{x}$, $\mathbf{V}$, and $\eta$, and the results are shown in Fig. \ref{fig:algcomparizon}.
  
   We ran the simulations on a computer server with 64 cores with 2.10GHz (Intel\textregistered \ Xeon\textregistered \ Gold 6130). We stopped the SCVI algorithm when $E[\eta]$ varied less than 0.5\% per iteration, and for the SVI, we considered a fixed number of 40 iterations. The VI algorithms were faster and scaled better, as shown in Table \ref{fig:timecomparizon}.

The posterior means of $\mathbf{x}$ obtained with the VI approximations were very similar to those obtained from Stan, while there was more discrepancy in the standard deviations. For mixing variables $V_i$, we can observe a negative bias when the true posterior mean of $V_i$ is larger than one and a positive bias otherwise.  In addition, the posterior standard deviations of $\mathbf{V}$ and $\eta$ were overall underestimated, which is typical for VI approximations based on KLD optimization \citep{minka2005divergence}. Also, the VI approximations underestimated the posterior mean of $\eta$, and compared to the mixing variables $V_i$, the posterior means of $\eta$ were more dispersed. Therefore, we should pay more attention to the mixing variables $V_i$ when judging the non-Gaussianity of a model, as the VI approximation is more reliable for these variables than for $\eta$.


\begin{figure}[htp]
\begin{tabular}{cc}
  \includegraphics[width=0.49\linewidth]{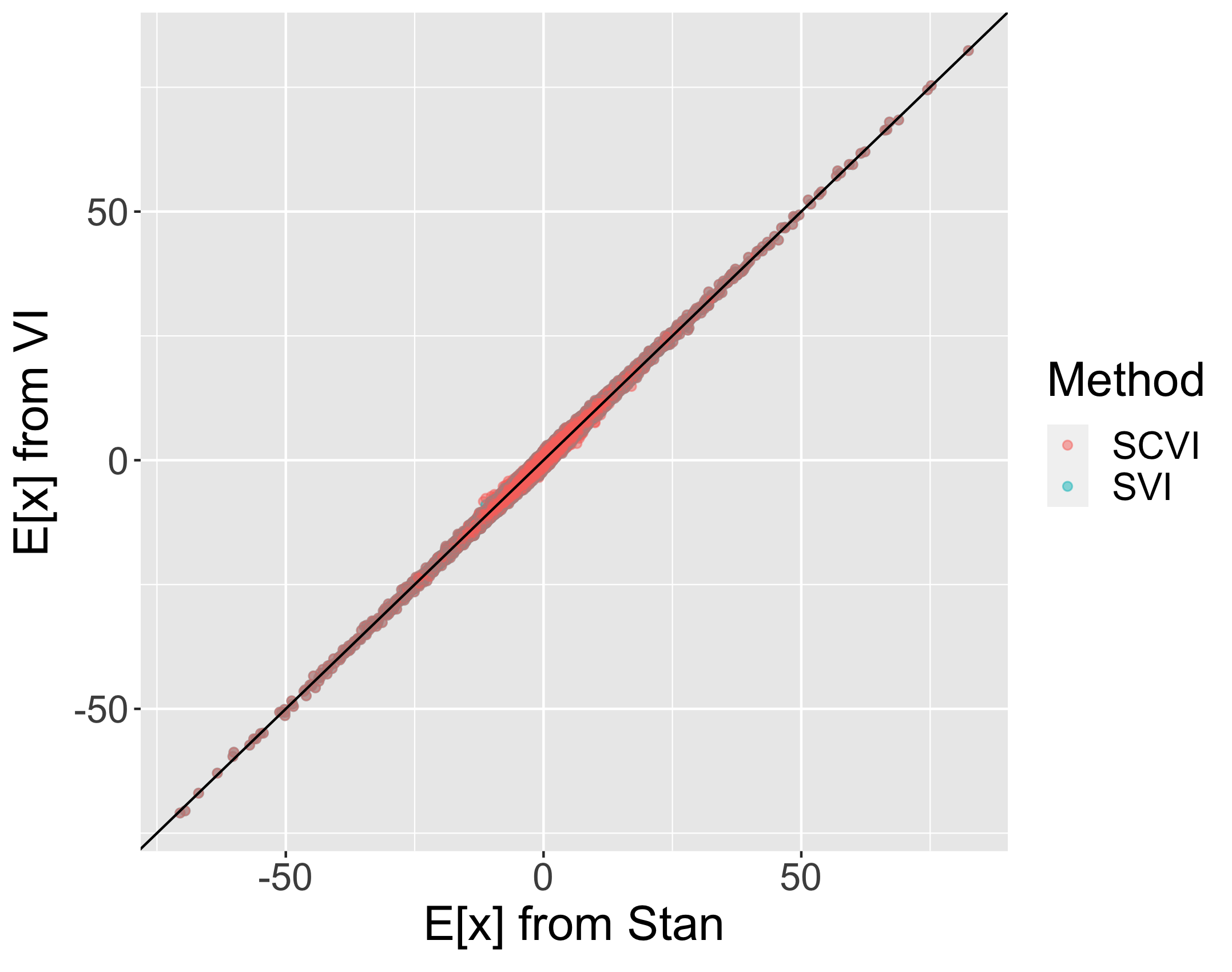} &  \includegraphics[width=0.49\linewidth]{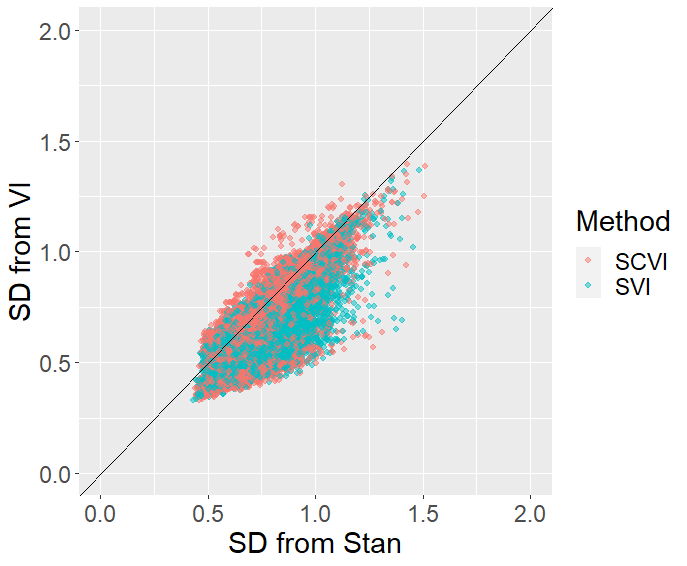}\\

  \includegraphics[width=0.49\linewidth]{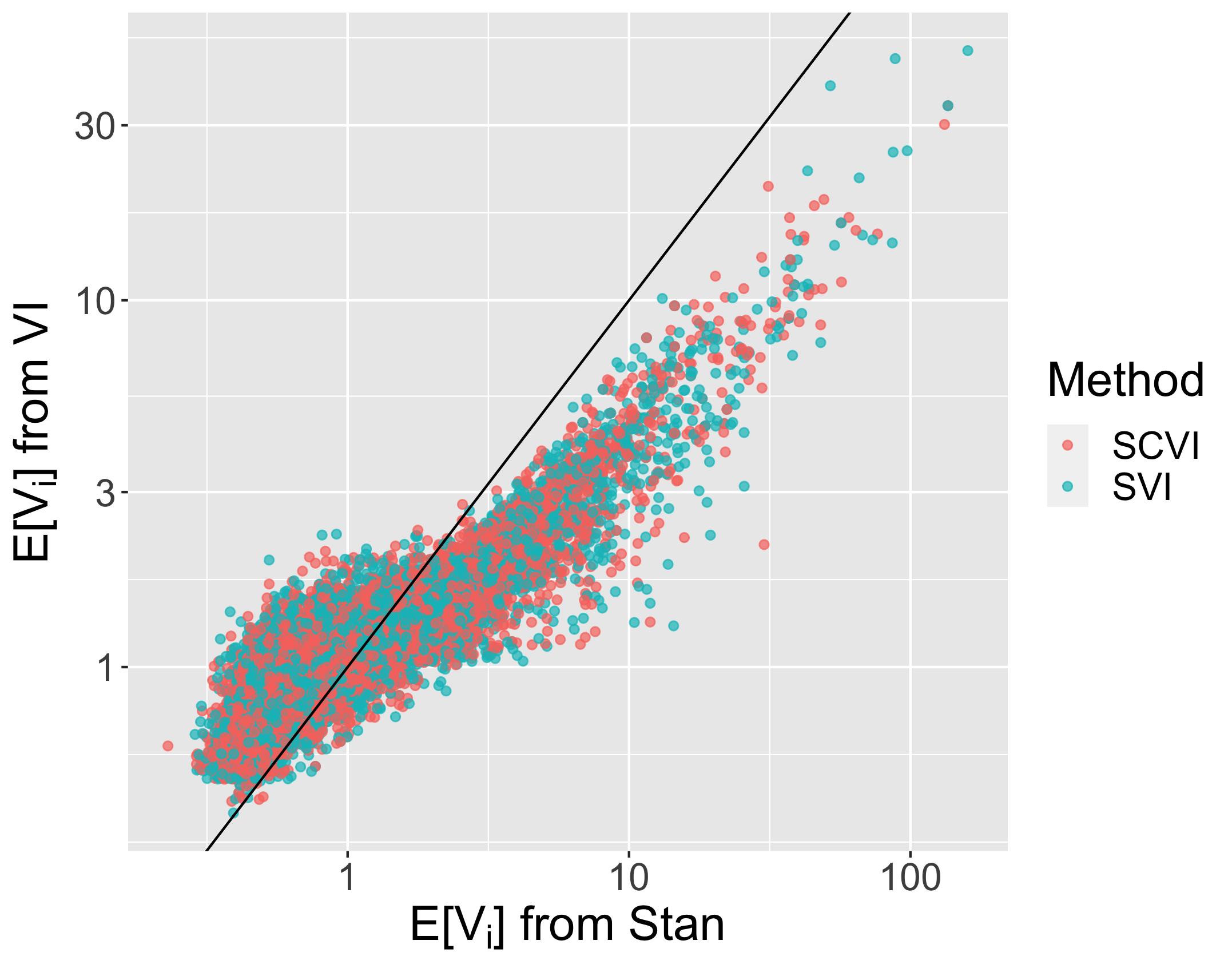} &   \includegraphics[width=0.49\linewidth]{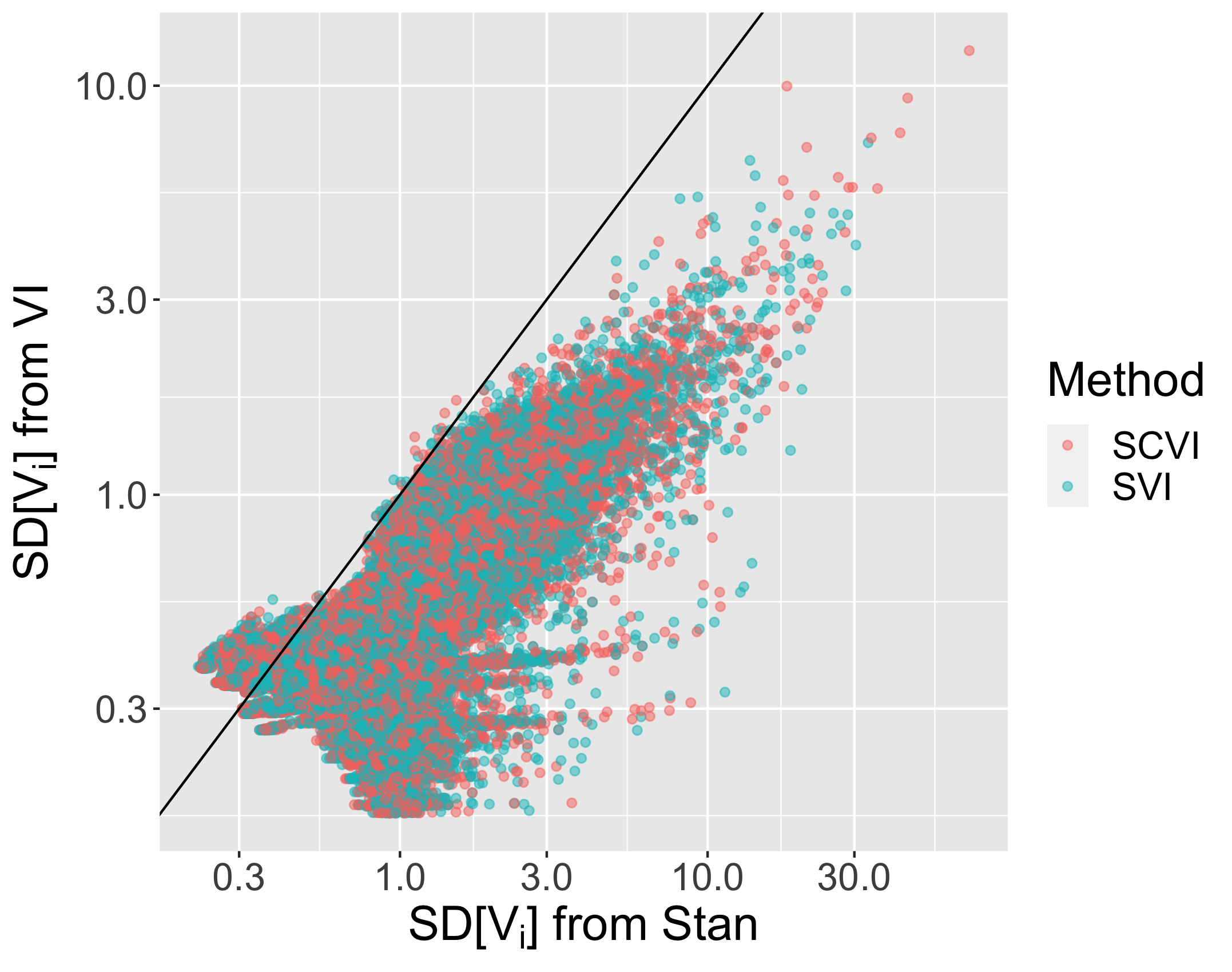} \\
  \includegraphics[width=0.49\linewidth]{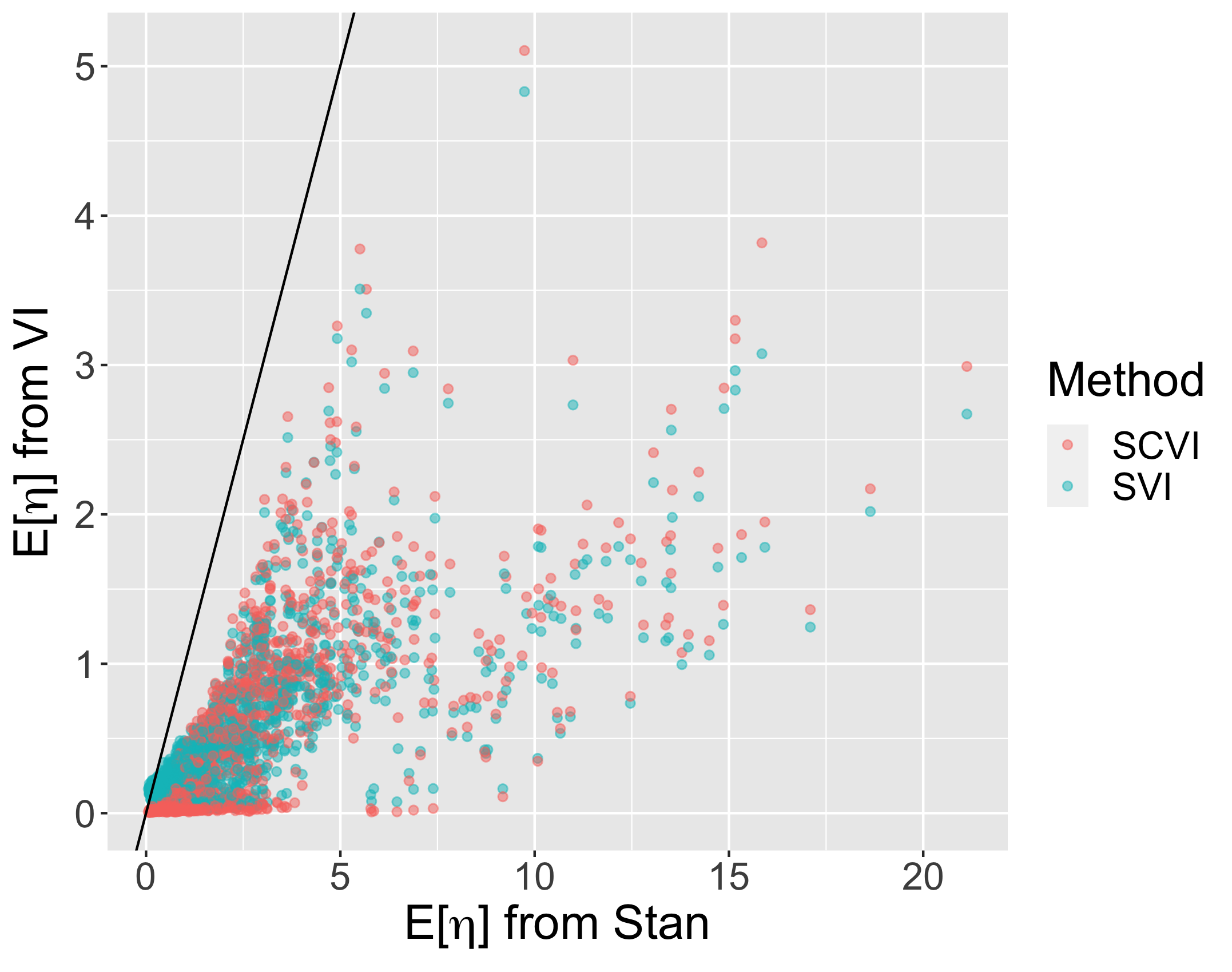} &    \includegraphics[width=0.49\linewidth]{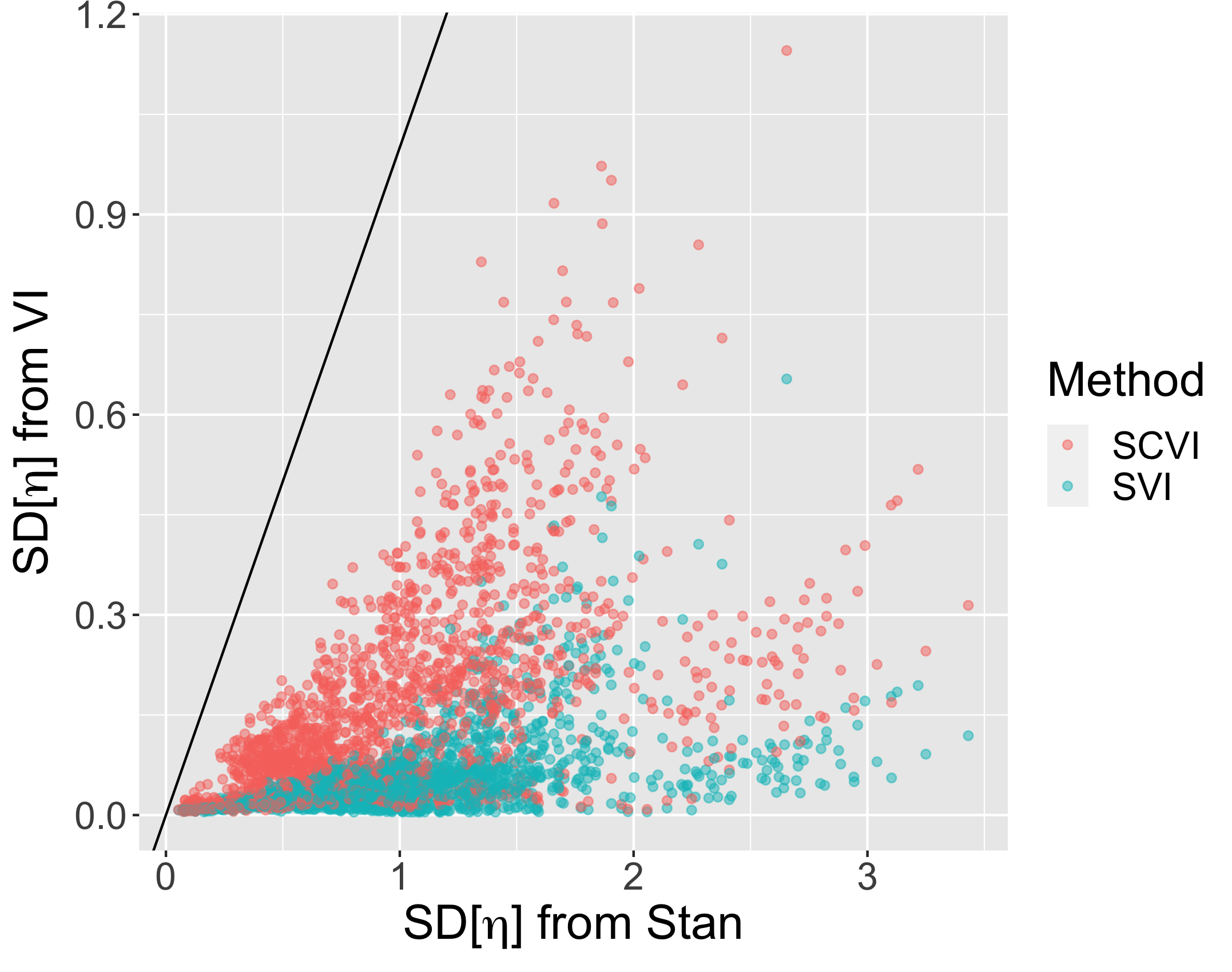} 
\end{tabular}
  \caption{The scatter plots show the posterior means and standard deviations of the latent field $\mathbf{x}$, mixing variables $V_i$ (in log scale) and parameter $\eta$ obtained from the VI approximations against their estimates obtained from Stan. We aggregated all simulation scenarios in these plots.}
  \label{fig:algcomparizon}
\end{figure}

\begin{table}[htp]
\centering
\begin{tabular}[t]{cccc}
	\textbf{Algorithm} & \textbf{N = 100} & \textbf{N = 500} & \textbf{N = 1000} \\ \hline
	Stan      & 49             & 3350           & 15267           \\
	SVI       & 85             & 148            & 257               \\
	SCVI      & 29             & 73             & 145             
\end{tabular} 
  \caption{Average elapsed time in seconds for each algorithm and dimension $N$.}
  \label{fig:timecomparizon}
\end{table}

\section{Applications} \label{sect:applications}
 
 
 
An LnGM is not always needed, and we need to ensure we are not overfitting the data  \citep{cabral2022controlling}. Therefore, in this section, we consider two applications. In the first one, there are outliers in the data, and there is an evident benefit in considering an LnGM. In the second application, the need for an LnGM is more ambiguous. The code for the applications can be found in the vignette of the \verb|ngvb| package.

 

\subsection{SAR model for areal data}

We study here areal data, which consists of the number of residential burglaries and vehicle thefts per thousand households ($y_i$) in 49 counties of Columbus, Ohio, in 1980. This dataset, shown in Fig. \ref{fig:crime}, can be found in the \emph{spdep} package in $R$. We observe several sharp variations in the crime rate of neighboring counties, which, with a deeper look at the data, do not seem fully explained by the available covariates. These sharp variations suggest that we could benefit from using a non-Gaussian model to account for the spatial effects. \cite{walder2020bayesian} analyzed this dataset using a non-Gaussian model to model the spatial dependency in \cite{walder2020bayesian}. We consider the same set of covariates as the previous authors and fit the following model:
\begin{equation} \label{eq:columbus}
y_{i}= \beta_0 + \beta_1 \mathrm{HV}_i + \beta_2 \mathrm{HI}_i +  \sigma_{\mathbf{x}}x_i + \sigma_{\epsilon}\epsilon_i,    
\end{equation}
where $\mathrm{HV}_i$ and $\mathrm{HI}_i$ are the average household value and household income for county $i$, $\mathbf{x}$ accounts for structured spatial effects, while $\epsilon_i \overset{i.i.d}{\sim} N(0,1)$ is an unstructured spatial effect.

We consider a simultaneous autoregressive (SAR) model \citep{besag1974spatial,wall2004close,ver2018relationship} for the spatially structured effects $\mathbf{x}$. The Gaussian version of this model can be built from the following relationship $
\mathbf{x} = \mathbf{B}\mathbf{x} + \sigma_{\mathbf{x}}\mathbf{Z}$ where each element of the random vector $\mathbf{x}$ corresponds to a county and $\mathbf{Z}$ is a vector of i.i.d. standard Gaussian driving noise. The matrix $\mathbf{B}$ causes simultaneous autoregressions of each random variable on its neighbors, where two regions are considered neighbors if they share a common border. For simplicity, we assume $\mathbf{B}=\rho\mathbf{W}$, where $\mathbf{W}$ is a row standardized adjacency matrix and $-1<\rho<1$ so that the resulting precision matrix is valid. We thus end up with the system $\mathbf{D}_{SAR}\mathbf{x} = \sigma_{\mathbf{x}}\mathbf{Z}$, where $\mathbf{D}_{SAR}=\mathbf{I}-\rho\mathbf{W}$. The equivalent model driven by NIG noise is then $\mathbf{D}_{SAR}\mathbf{x} = \sigma_{\mathbf{x}}\mathbf{\Lambda}$, where $\mathbf{\Lambda}$ is i.i.d. standardized NIG noise. 

The model in  \eqref{eq:columbus} is thus an LnGM, and we fitted this model using the $\verb|ngvb|$ function with the SCVI algorithm. We also fitted the reciprocal Gaussian model for comparison in INLA. The Bayes factor between the LnGM and LGM was around 80000. We show the posterior means of the mixing variables in Fig. \ref{fig:crime} and the posterior summaries of the other parameters in Table \ref{table:crime}. We can observe in Fig. \ref{fig:crime} two pronounced outlier counties, and the consequence of downweighting these counties on the analysis was the reduction of the relationship between household value and crime rate.

\begin{figure}[h]
\begin{tabular}{cc}
  \includegraphics[width=0.49\linewidth, height = 5cm]{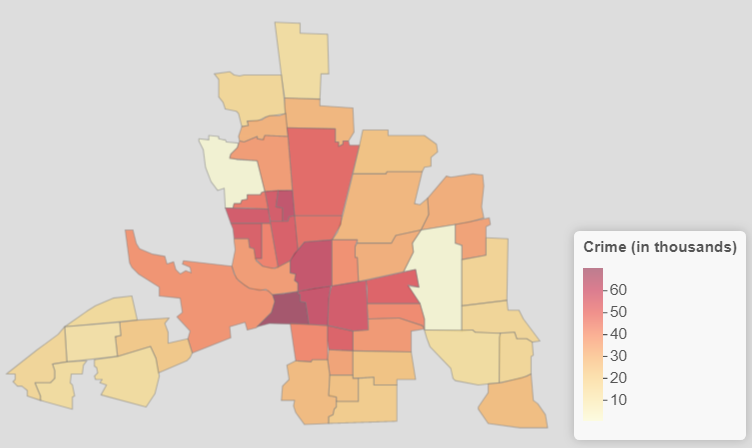} &   \includegraphics[width=0.49\linewidth, height = 5cm]{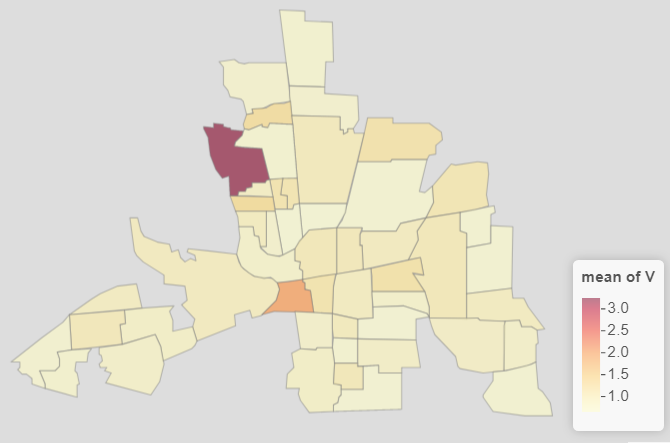} \\
\end{tabular}
  \caption{Crime rates in Columbus (left) per thousand households and posterior means of the mixing variables $V_i$ (right).}
  \label{fig:crime}
\end{figure}

 \begin{table}[h]
 \centering
\begin{tabular}{ccc}
\textbf{Parameter}                     & \textbf{LGM}          & \textbf{LnGM}              \\ \hline

Intercept                     & 60.065 (46.860, 72.010)                   & 58.873 (47.076, 69.499)                  \\
Household Value                   & -0.303 (-0.489, -0.117)                   & - 0.169 (-0.320, -0.018)                  \\
Household Income               & -0.948 (-1.691,-0.224)                    & -1.182 (-1.793, -1.177)          \\
$\sigma_{\boldsymbol{\epsilon}}$ & 0.008 (0.004,0.015 )                      & 0.010 (0.004, 0.023)                     \\
$\sigma_{\mathbf{x}}$           & 10.077 (8.165,12.439 )                    & 9.730 (7.917, 11.980)                    \\
$\rho$                        & 0.566 (0.242 ,0.839) & 0.648 (0.385, 0.860) \\
$\eta$ & - &    0.679 (0.100, 1.702)
\end{tabular}
\caption{Posterior mean and credible intervals (based on the 97.5\% quantiles) of the model's parameters in  \eqref{eq:columbus}.}
\label{table:crime}
\end{table}


\subsection{Model with random slope and intercept} \label{sect:random.slople.intercept}
 The data comes from an orthodontic study reported by \cite{potthoff1964generalized}. The response variable shown in Fig. \ref{fig:app1data} is the distance in millimeters between the pituitary and the pterygomaxillary fissure, measured for 11 girls and 16 boys at ages 8, 10, 12, and 14. 

\begin{figure}[htp]
  \centering
  \includegraphics[width=0.9\linewidth]{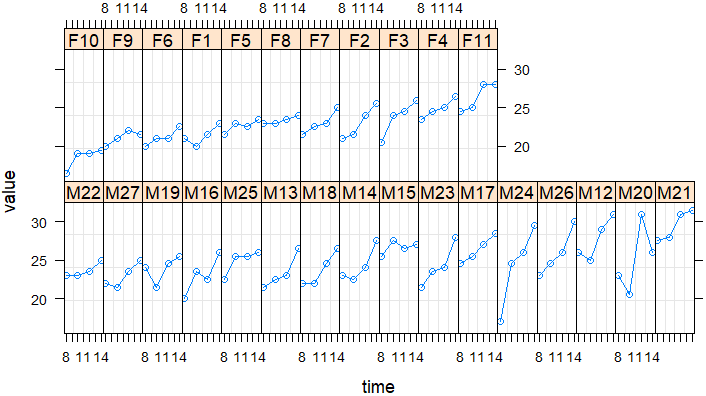}
  \caption{Trellis display showing the distance in millimeters between the pituitary and the pterygomaxillary fissure for girls (first row) and boys (second row).}
  \label{fig:app1data}
\end{figure}

 The data suggests that the intercept and slope vary with the subject. Thus, \cite{pinheiro2001efficient} proposed the following linear mixed-effects model to describe the response growth with age:
\begin{equation}\label{eq:app1}
y_{i j}=\beta_0+\delta_0 I_i(F) +\left(\beta_1+\delta_1 I_i(F)\right) t_j + b_{0 i}+b_{1 i} t_j+\epsilon_{i j},
\end{equation}
where $y_{i j}$ denotes the response for the $i$th subject at age $t_j$, $i=1, \ldots, 24$ and $j=1, \ldots, 4;$  $\beta_0$ and $\beta_1$ denote, respectively, the intercept and the slope fixed effects for boys; $\delta_0$ and $\delta_1$ denote, respectively, the difference in intercept and slope fixed effects between girls and boys; $I_i(F)$ denotes an indicator variable for females; $\mathbf{b}_i=\left(b_{0 i}, b_{1 i}\right)$ is the random effects vector for the $i$ th subject; and $\epsilon_{i j}$ is the within-subject error.

There are three possible sources of outliers: the intercept, the slope, and the within-subject error, and non-Gaussianity could be necessary for just one of these three components. However, \cite{pinheiro2001efficient} considered a $t$-Student model with the same degrees of freedom parameter for all components. Therefore, if there are only outliers in one component, their construction will add unneeded extra flexibility to the other two components. 

 
Instead, we consider separate non-Gaussian models for the random intercept and random slope, where each component has its non-Gaussianity parameters $\eta_0$ and $\eta_1$. \cite{asar2020linear} describes non-Gaussian models of this type to the general family of linear mixed-effects models. We set $b_{0 i} \overset{i.i.d}{\sim} \mathrm{NIG}(0,\sigma_0^2,\eta_0)$ independently from $b_{1 i} \overset{i.i.d}{\sim} \mathrm{NIG}(0,\sigma_1^2,\eta_1)$, where $\mathrm{NIG}(\mu,\sigma^2,\eta_2)$, stands for a NIG distribution with mean $\mu$, variance $\sigma^2$ and non-Gaussianity parameter $\eta$.  

We finally considered normally distributed errors $\epsilon_{ij} \overset{i.i.d}{\sim} N(0,\sigma_\epsilon^2)$. The Gaussian version of this model can be implemented in INLA, and we used the \verb|ngvb| function with the SCVI algorithm to extend this model to non-Gaussianity. Fig. \ref{fig:postV2} shows each component's posterior means and credible intervals of the mixing variables $V_i$. We also show the results obtained from the Gibbs sampler in \verb|ngvb|.





The mixing variables $V_i$ for the random intercept component had the largest mean for subjects F10 and M21, and if we look back at Fig. \ref{fig:app1data}, the orthodontic distance in the first measurement was unusually small and large for these subjects, respectively. Conversely, for the random slope component, we found no significant departure from Gaussianity in any subject since the posterior means of the mixing variables were very close to 1, suggesting that a non-Gaussian model was not needed for this component. Posterior summaries of the parameters of the LGM and LnGM are found in Table \ref{table:randomslopeinter}. None of the parameters changed substantially, and the Bayes factor was around 7, slightly favoring the LnGM. 





\begin{figure}[H]
  \centering
  \includegraphics[width=\linewidth]{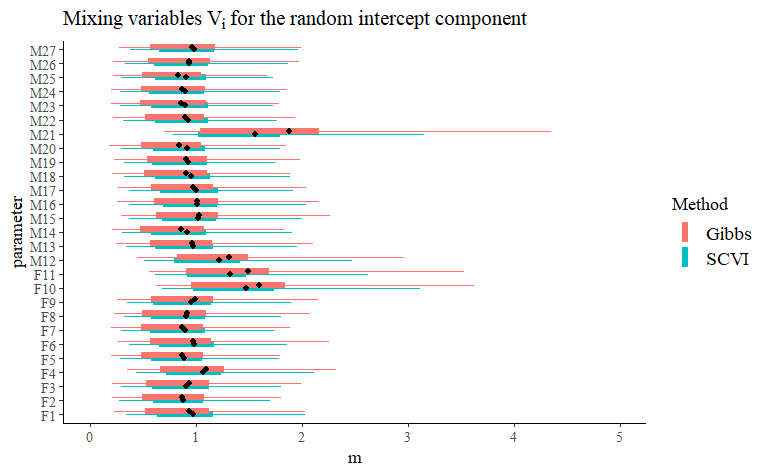} \\
  \includegraphics[width=\linewidth]{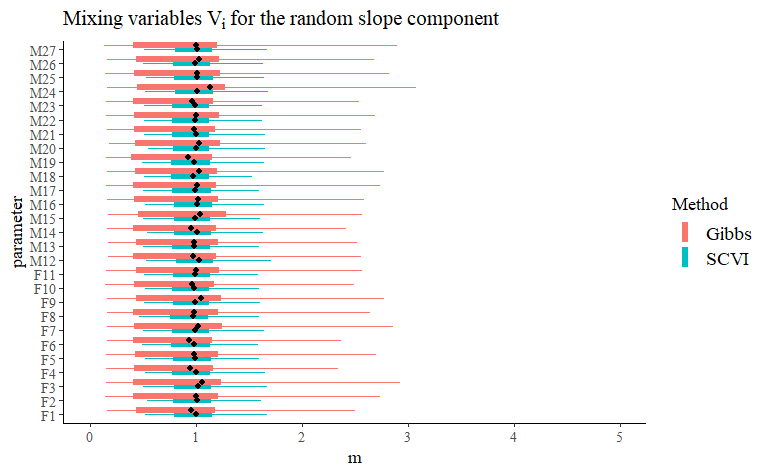}
  \caption{Boxplots of the mixing variables $V_{i}$ for the random intercepts (top) and random slopes (bottom) using a Gibbs sampler and the SCVI algorithm. Each bar contains the posterior means (black dot) and central intervals based on 50\% and 90\% quantiles.}
  \label{fig:postV2}
\end{figure}

\begin{table}[htp] 
\centering
\begin{tabular}{ccc}
\textbf{Parameters}          & \textbf{LGM} & \textbf{LnGM} \\ \hline
$\sigma_\epsilon$  & 1.375 (1.185, 1.590)             & 1.373 (1.182, 1.602 )     	   \\
$\sigma_0$         & 1.816 (1.387, 2.421)             & 1.839 (1.336, 2.485)        \\
$\sigma_1$         & 0.014 (0.005, 0.042)             & 0.012 (0.004, 0.030)         \\
$\beta_0$          & 22.615 (21.578, 23.653)           & 22.535 (21.782, 23.288)       \\
$\delta_0$         & -0.796 (-2.679, 1.086)           & -0.599 (-2.380, 1.182)       \\
$\beta_1$          & 1.569 (1.264, 1.873)             & 1.568 (1.266, 1.870)        \\
$\delta_1$         & -0.610 (-1.087, -0.133)          & -0.610 (-1.083, -0.136)    \\
\end{tabular}
\caption{Posterior means and 97.5\% credible intervals for the parameters of the LGM and LnGM models of the orthodontic study.}
\label{table:randomslopeinter}
\end{table}






\section{Discussion} \label{sect:discussion}

We derived in this paper two variational inference algorithms to approximate a generic class of latent non-Gaussian models. The approximation leads to an LGM that downweights outlier events of the latent process and provides a means to obtain robustness and diagnostics. Robustness is obtained because we safeguard against deviations from the latent Gaussian assumption. We also constructed diagnostic plots based on $q(\mathbf{V})$ (see Figs.  \ref{fig:progression}, \ref{fig:crime}, and \ref{fig:postV2}), which allow us to identify where those deviations occur.

Other variational Bayes algorithms could be explored to obtain more accurate approximations, for instance, those based on setting a fixed form distribution on  $q(\mathbf{V})$ \citep{tran2021practical}. Finally, for routine use by non-experts, we developed the \verb|ngvb| package, which facilitates Bayesian inference of LnGMs. LnGMs can be fitted by just adding a line of code to pre-existing implementations of LGMs with R-INLA, and future work includes making the package compatible with more models.

\bibliographystyle{Chicago}

\bibliography{main}

\begin{appendices}

The appendices contain the list of models that can be extended to non-Gaussianity according to section \ref{sect:ngproc}, the proofs of the theorems, several extensions of the main results, and more details about the Gibbs sampler in the \verb|ngvb| package. 

\section{Non-Gaussian models}\label{app:ng}

We list in Table \ref{table:D} several models and their corresponding dependency matrices $\mathbf{D}$. More details about the construction of these matrices can be found in \cite{cabral2022controlling} and references therein. We include models with independent and identically distributed (i.i.d.) elements, random walks of order 1 and 2 (RW1 and RW2), autoregressive processes (AR), Mátern models defined through the stochastic partial differential equation (SPDE) approach, intrinsic conditional autoregressive (ICAR) and simultaneous autoregressive (SAR) models. We note that different parameterizations could have been used for some models. Finally, we list some improper models, i.e., whose precision matrices are not full rank, namely the RW1, RW2, and ICAR (intrinsic conditional autoregressive) models. These models are discussed in detail in \cite{rue2005gaussian}.

\begin{table}[htp]
\begin{tabular}{ccc}
\textbf{Model}  & \textbf{Matrix D} & \textbf{Notes} \\ \hline \\
i.i.d.          &     $\mathbf{I}_{N\times N}$              &                \\ \\

RW1             &    
\begin{math}
\begin{bmatrix}
-1 & 1 & 0 & 0 & \ \\
0 & -1 & 1 & 0 & \ \\
0 & 0 & -1 & 1 & \ \\
  &   &    &   & \ddots 
\end{bmatrix}_{(N-1)\times N}
\end{math}
           &                  \\ \\

RW2             &        \begin{math}
\begin{bmatrix}
1 & -2 & 1 & 0 & \ \\
0 & 1 & -2 & 1 & \ \\
0 & 0 & 1 & -2 & \ \\
  &   &    &   & \ddots \\
\end{bmatrix}_{(N-2)\times N} \end{math}              &            \\ \\

AR1             &       \begin{math}\begin{bmatrix}
\sqrt{1-\rho^2} & 0 & 0 & 0 &  \ \\
-\rho & 1 & 0 & 0 & \ \\
0 & -\rho & 1 & 0 & \ \\
  &   &    &   & \ddots \\
\end{bmatrix}_{N \times N}  \end{math}                    &   \makecell{$|\rho|<1$ is the autocorrelation \\ coefficient}  \\ \\

Higher order AR &     \makecell{Toeplitz matrix where the \\ lower diagonal contains the \\ autocorrelation parameters. }                &                \\ \\
\makecell{ Matérn models \\ via SPDE method} &    $\kappa^2 \mathbf{C} + \mathbf{G}$               &  \makecell{$\kappa$ is the spatial range parameter. \\ The matrices $\mathbf{C}$ and $\mathbf{G}$ are \\ in \cite{lindgren2011explicit} }.             \\ \\
ICAR             &        \begin{math}\begin{bmatrix}
1 & 1 & 0 & 0 \\
1 & 0 & 1 & 0 \\
0 & 1 & 0 & 1   \\
\end{bmatrix} \end{math}              &  \makecell{The following nodes are neighbours \\ $1 \leftrightarrow 2$, $1 \leftrightarrow 3$, and  $2 \leftrightarrow 4$  \\ Each row contains one \\ neighbouring relationship. \\ We place the value 1 at the \\ nodes' indices.}            \\ \\
SAR             &    $\mathbf{I} - \rho \mathbf{W}$                 &  \makecell{$\mathbf{W}$ is the row standardized \\ adjacency matrix, $|\rho|<1$ is the \\ autocorrelation coefficient. }
\end{tabular}
\caption{Several models and their corresponding dependency matrices $\mathbf{D}$.}
\label{table:D}
\end{table}

\newpage
\section{Proof of the theorems}

We present here the proofs of the theorems shown in the paper. To simplify the notation we consider $\mathbf{D} = \mathbf{D}(\boldsymbol{\theta}_2)$ and $\mathbf{V}^{\circ -1}$ stands for the element-wise inverse $[V_1^{-1}, \dotsc, V_N^{-1}]^\top$. The multivariate normal distribution is defined through its mean $\mathbf{m}$ and precision matrix $\mathbf{Q}$, and:
\begin{equation}\label{eq:useful2}
\pi_{\mathrm{Normal}}(\mathbf{x};\mathbf{m} = \mathbf{0}, \mathbf{Q} = \mathbf{D}^\top \mathrm{diag}(\mathbf{V}^{\circ -1}) \mathbf{D}) = \frac{\det(\mathbf{D})}{(2\pi)^{N/2}}\prod_{i=1}^N V_i^{-1/2} \exp\left(-\frac{1}{2V_i} [\mathbf{D}\mathbf{x}]^2_i \right).
\end{equation}
Additionally, we will make use of the following equality:
\begin{equation}\label{eq:useful1}
    \pi_{\mathrm{GIG}}(x; p_1, a_1, b_1)\pi_{\mathrm{GIG}}(x; p_2, a_2, b_2) = \pi_{\mathrm{GIG}}(x; p_1 + p_2 -1, a_1+a_2, b_1+b_2).
\end{equation}

\subsection{Proof of Theorem 1} \label{app:proof1}
\begin{proof}
From  \eqref{eq:CAVI} of the main paper, the solution to the variational problem is the system:

\[
\log q(\mathbf{x}, \boldsymbol{\theta}) \propto E_{q(\mathbf{V})q(\eta)}(\log \pi(\mathbf{x}, \boldsymbol{\theta}, \mathbf{V},\eta|\mathbf{y})),
\]

\[
\log q(\mathbf{V}) \propto E_{q(\mathbf{x}, \boldsymbol{\theta})q(\eta)}(\log \pi(\mathbf{x}, \boldsymbol{\theta}, \mathbf{V},\eta|\mathbf{y})),
\]

\[
\log q(\eta) \propto E_{q(\mathbf{x}, \boldsymbol{\theta}) q(\mathbf{V})}(\log \pi(\mathbf{x}, \boldsymbol{\theta}, \mathbf{V},\eta|\mathbf{y})).
\]

We find the surrogate distributions $q(\mathbf{x}, \boldsymbol{\theta}), q(\mathbf{V}), q(\eta)$ in the following three steps. For this, It is useful to decompose the posterior joint density into 6 parts:   
\[\log \pi(\mathbf{x}, \boldsymbol{\theta}, \mathbf{V},\eta|\mathbf{y}) = \log \pi(\mathbf{y} | \mathbf{x}, \boldsymbol{\theta}_1) + \log \pi(\mathbf{x} | \mathbf{V}, \boldsymbol{\theta}_2) + \log \pi(\mathbf{V}|\eta) + \log \pi(\boldsymbol{\theta}) + \log \pi(\eta) - \log \pi(\mathbf{y}).\]

\vspace{1cm}

\noindent
{ \bf Step 1: Surrogate density $q(\mathbf{x}, \theta)$.} For $q(\mathbf{x}, \boldsymbol{\theta})$ we have:
$$
\log q(\mathbf{x}, \boldsymbol{\theta}) \propto \log \pi(\mathbf{y} | \mathbf{x}, \boldsymbol{\theta}_1) + E_{q(\mathbf{V})}[\log \pi(\mathbf{x} | \mathbf{V}, \boldsymbol{\theta}_2)] + \log \pi(\boldsymbol{\theta}).
$$
The second term $E_{q(\mathbf{V})}[\log \pi(\mathbf{x} | \mathbf{V}, \boldsymbol{\theta}_2)] \propto -\frac{1}{2} \mathbf{x}^\top \mathbf{D}^\top \mathrm{diag}(\mathbf{\tilde{V}}) \mathbf{D} \mathbf{x} $ is proportional to the log-density of a Gaussian distribution for $\mathbf{x}$ with mean $\mathbf{0}$ and precision matrix $\tilde{\mathbf{Q}}=\mathbf{D}^\top\mathrm{diag}(\mathbf{\tilde{V}})\mathbf{D}$, where $\mathbf{\tilde{V}} = E_{q(\mathbf{V})}(\mathbf{V}^{\circ -1})$. Therefore the surrogate density is:
\begin{equation*} \label{eq:VI1}
    q(\mathbf{x}, \boldsymbol{\theta}_1) \propto \pi(\mathbf{y} | \mathbf{x}, \boldsymbol{\theta}) \pi_{\mathrm{Normal}}(\mathbf{x}; \mathbf{0}, \mathbf{\tilde{Q}})\pi(\boldsymbol{\theta}),
\end{equation*}
 which is the posterior distribution of ${\mathrm{LGM} \{ \pi(\mathbf{y}|\mathbf{x}, \boldsymbol{\theta}_1), \mathbf{0}, \tilde{\mathbf{Q}}, \pi(\boldsymbol{\theta}) \} }$.

\vspace{1cm}

\noindent
{ \bf Step 2: Surrogate density $q(\mathbf{V})$.} For $q(\mathbf{V})$ we have: $$\log q(\mathbf{V}) \propto  E_{q(\mathbf{x},\boldsymbol{\theta}_2)}[\log\pi(\mathbf{x} | \mathbf{V}, \boldsymbol{\theta}_2)] + E_{q(\eta)}[\log\pi(\mathbf{V}|\eta)].$$ From the multivariate normal expansion in  \eqref{eq:useful2} the first term in the sum simplifies to:
\begin{align*}
    E_{q(\mathbf{x},\boldsymbol{\theta}_2)}[\log\pi(\mathbf{x} | \mathbf{V}, \boldsymbol{\theta}_2)] &\propto \sum_{i=1}^N -\frac{1}{2} \log(V_i) -\frac{1}{2} V_i^{-1} E_{q(\mathbf{x},\boldsymbol{\theta}_2)}\left([\mathbf{D}\mathbf{x}]_i^2\right)  \\ &= \sum_{i=1}^N \log \pi_{\mathrm{GIG}}\left(V_i; 1/2, 0,  E_{q(\mathbf{x},\boldsymbol{\theta}_2)}\left([\mathbf{D}\mathbf{x}]_i^2\right) \right),
\end{align*}
The second term can also be expressed as a sum of GIG log densities:
\begin{align*}
    \log E_{q(\eta)}[\log\pi(\mathbf{V}|\eta)] &\propto \sum_{i=1}^N -\frac{3}{2}\log(V_i)  -\frac{V_i}{2}\left(E_{q(\eta)}[\eta^{-1}]\right) -\frac{1}{2V_i}\left(h_i^2E_{q(\eta)}[\eta^{-1}]\right) \\
    &= \sum_{i=1}^N\log \pi_{\mathrm{GIG}}\left(V_i; -1/2, E_{q(\eta)}[\eta^{-1}], h_i^2E_{q(\eta)}[\eta^{-1}]  \right).
\end{align*}
Therefore, using  \eqref{eq:useful1}, the sum of all GIG log-densities simplifies to: 
\begin{align*}
\log q(\mathbf{V}) &= \sum_{i=1}^N \log \pi_{\mathrm{GIG}}\left(V_i; -1, E_{q(\eta)}[\eta^{-1}], h_i^2E_{q(\eta)}[\eta^{-1}] + E_{q(\mathbf{x},\boldsymbol{\theta}_2)}\left([\mathbf{D}\mathbf{x}]_i^2\right) \right).
\end{align*}

\vspace{1cm}

\noindent
{ \bf Step 3: Surrogate density $q(\eta)$.} For $q(\eta)$ we have: $\log q(\eta) \propto  E_{q(\mathbf{V})}[\log \pi(\mathbf{V} | \eta)] + \log\pi(\eta)$. The first term of the sum is:
\begin{align*}
    E_{q(\mathbf{V})}[\log \pi(\mathbf{V} | \eta)] &\propto -\frac{N}{2}\log(\eta) - \frac{1}{2\eta} E_{q(\mathbf{V})}\left[ \sum_{i=1}^N \frac{(V_i-h_i)^2}{V_i}\right] \\
    &= \log \pi_{\mathrm{GIG}}\left(\eta; -N/2 + 1, 0, E_{q(\mathbf{V})}\left[ \sum_{i=1}^N\frac{(V_i-h_i)^2}{V_i}\right]  \right)
\end{align*}
We consider an exponential prior for the parameter $\eta$, and since the exponential distribution is a special case of the GIG distribution, the second term $\log\pi(\eta)$ can be expressed as
\begin{align*}
    \log\pi(\eta) = \log \pi_{\mathrm{GIG}}(\eta; 1, 2\alpha_\eta, 0).
\end{align*}
Leveraging again on  \eqref{eq:useful1} the sum of two GIG log-densities simplifies to: 
\begin{align*}
\log q(\eta) = \log \pi_{\mathrm{GIG}}\left(\eta; -N/2 + 1, 2\alpha_\eta, \sum_{i=1}^N E_{q(V_i)}[V_i] -2h_i + h_i^2 E_{q(V_i)}[V_i^{-1}] \right).
\end{align*}

\end{proof}

\subsection{Proof of Theorem 2 } \label{app:proof2}

\begin{proof}
\sloppy We now find the surrogate density $q(\mathbf{x}, \boldsymbol{\theta}, \mathbf{V}) = q(\mathbf{x}, \boldsymbol{\theta}) q(\mathbf{V})$ that minimises $\mathrm{KLD}(q(\mathbf{x}, \boldsymbol{\theta}, \mathbf{V})|\pi(\mathbf{x}, \boldsymbol{\theta}, \mathbf{V}|\mathbf{y}))$. The starting point is again the system:
\begin{align*}
\log q(\mathbf{x}, \boldsymbol{\theta}) &\propto E_{q(\mathbf{V})}[\log \pi(\mathbf{x}, \boldsymbol{\theta}, \mathbf{V}|\mathbf{y})], \\ \\
\log q(\mathbf{V}) &\propto E_{q(\mathbf{x}, \boldsymbol{\theta})}[\log \pi(\mathbf{x}, \boldsymbol{\theta}, \mathbf{V}|\mathbf{y})]. 
\end{align*}

The surrogate distribution $q(\mathbf{x}, \boldsymbol{\theta})$ will be the same as in Theorem \ref{theo:1}, which can be shown by repeating the same calculations done in step 1 of the proof of Theorem \ref{theo:1}. For $ q(\mathbf{V})$ we have:
$$
\log q(\mathbf{V})\propto  E_{q(\mathbf{x},\boldsymbol{\theta}_2)}[\log \pi(\mathbf{x} |\mathbf{V},\boldsymbol{\theta}_2)] +  \log \pi(\mathbf{V}),
$$
and so $q(\mathbf{V})= \exp( E_{q(\mathbf{x},\boldsymbol{\theta}_2)}[\log \pi(\mathbf{x} |\mathbf{V},\boldsymbol{\theta}_2)] ) \pi(\mathbf{V})$.
The first term is given by:
\begin{equation}\label{eq:surrogateVproof}
 \exp(E_{q(\mathbf{x}, \boldsymbol{\theta})}[\log \pi(\mathbf{x} |\mathbf{V},\boldsymbol{\theta}_2)]) \propto \prod_{i=1}^N V_i^{-1/2} \exp\left(-\frac{1}{2 V_i} E_{q(\mathbf{x}, \boldsymbol{\theta}_2)}([\mathbf{D}\mathbf{x}]_i^2) \right),    
\end{equation}
where we used  \eqref{eq:useful2} to expand $\log \pi(\mathbf{x} |\mathbf{V},\boldsymbol{\theta}_2)$. The second term is:
\begin{align}\label{eq:surrogateVproof2}
\pi(\mathbf{V}) &= \int_{0}^{\infty} \left(\prod_{i=1}^N \pi_{\mathrm{GIG}}(V_i; -1/2, \eta^{-1}, h_i^2\eta^{-1}) \right) \pi_{\mathrm{Exp}}(\eta; \alpha_\eta) d\eta, \\
                &\propto \int_{0}^{\infty} \left(\prod_{i=1}^N \eta^{-1/2} V_i^{-3/2}\exp\left( -\frac{1}{2\eta}  \frac{(V_i-h_i)^2}{V_i}\right)    \right)    \exp\left(-\alpha_\eta \eta \right) d\eta. \nonumber
\end{align}

By joining both terms, we get the following:
\begin{align}\label{eq:surrogateVproof3}
q(\mathbf{V}) &\propto \int_{0}^{\infty} \prod_{i=1}^N \left(\frac{ \exp\left(-\frac{1}{2 V_i} E_{q(\mathbf{x}, \boldsymbol{\theta}_2)}([\mathbf{D}\mathbf{x}]_i^2) -\frac{1}{2\eta} \frac{(V_i-h_i)^2}{V_i} \right)}{\eta^{1/2} V_i^{2}} \right) \exp\left(-\alpha_\eta \eta  \right) d\eta. 
\end{align}

Solving the previous integral gives us a high dimensional unnormalized pdf for $q(\mathbf{V})$, which is hard to analyze and sample. However, there is an alternative representation for $q(\mathbf{V})$ given by $\int_0^\infty (\prod_{i=1}^N q(V_i|\eta)) q(\eta) d\eta$, which means that we can get samples from $q(\mathbf{V})$ by first sampling from $q(\eta)$ and then sampling from $ q(V_i|\eta)$ for each $V_i$. We can see that  \eqref{eq:surrogateVproof3} is expressed by $q(\mathbf{V}) =  \int_0^\infty (\prod_i p(V_i,\eta)) p(\eta) d\eta$, where $p(\eta) = e^{-\alpha_\eta \eta}$, and
\begin{align*}
p(V_i,\eta) &=  \eta^{-1/2} V_i^{-2} \exp\left(-\frac{1}{2 V_i} d_i -\frac{1}{2\eta} \frac{(V_i-h_i)^2}{V_i} \right), \ \ \ \ \ \ d_i = E_{q(\mathbf{x}, \boldsymbol{\theta}_2)}([\mathbf{D}\mathbf{x}]_i^2)\\
            &= V_i^{-2} \exp\left(-\frac{V_i}{2}\eta^{-1} -\frac{1}{2 V_i} (d_i+ h_i^2\eta^{-1}) \right) \eta^{-1/2}  e^{h_i\eta^{-1}} \\
            &=   \pi_{\mathrm{GIG}}(V_i; -1, \eta^{-1}, d_i + h_i^2\eta^{-1}) \eta^{-1/2}e^{h_i\eta^{-1}} \frac{2 K_{-1}(\sqrt{\eta^{-1}(d_i+h_i^2\eta^{-1})})}{\sqrt{d_i\eta + h_i^2}}.
\end{align*}
Therefore we have:
\[
q(V_i|\eta) \sim \pi_{\mathrm{GIG}}(V_i; -1, \eta^{-1}, d_i + h_i^2\eta^{-1}),
\]
and the remaining terms in the integrand of  $q(\mathbf{V}) =  \int_0^\infty (\prod_i p(V_i,\eta)) p(\eta) d\eta$ are aggregated in:
\[
q(\eta) \propto \eta^{-N/2}e^{\eta^{-1}(\sum_{i=1}^N h_i) - \alpha_\eta \eta} \prod_{i=1}^N\left( \frac{ K_{-1}(\sqrt{\eta^{-1}(d_i+h_i^2\eta^{-1})}) }{\sqrt{d_i\eta + h_i^2}}\right).
\]

\end{proof}

\section{Extensions} \label{sect:extensions}

In some applications, the model is comprised of several latent components. For instance, section \ref{sect:random.slople.intercept} considers a model with a random intercept and a random slope component. Also, one might be interested in using other driving noise distributions, such as the heavy-tailed $t$-Student distribution. We discuss these extensions here.


\subsection{Several latent components}

When the latent field comprises several independent random effects, the modifications to Theorems \ref{theo:1} and \ref{theo:2} are straightforward. Consider $\mathbf{x} = [\mathbf{x}_1,\mathbf{x}_2]^\top$, where $\mathbf{x}_1$ models temporal dependence (with dependency matrix $\mathbf{D}_1$) while $\mathbf{x}_2$ models spatial dependence (with dependency matrix $\mathbf{D}_2$). Let the first component have dimension $N_1$, mixing variables $\mathbf{V}_1$, non-Gaussianity parameter $\eta_1$ (whose exponential prior has rate $\alpha_{\eta_1}$), and predefined constants $\mathbf{h}_1$, and likewise for the second model component. Also, let $x_{1,i}$ refer to the element $i$ of the first component $\mathbf{x}_1$. \sloppy For the SVI approximation, we consider the surrogate $q(\mathbf{x}_1,\mathbf{x}_2,\boldsymbol{\theta},\mathbf{V}_1,\mathbf{V}_2,\eta_1,\eta_2) = q(\mathbf{x}_1,\mathbf{x}_2,\boldsymbol{\theta})q(\mathbf{V}_1)q(\mathbf{V}_2)q(\eta_1)q(\eta_2)$. The surrogate $q(\mathbf{x}_1,\mathbf{x}_2,\boldsymbol{\theta})$ is still the posterior distribution of an LGM whose precision matrix is $\mathbf{Q} \propto \mathbf{D}^\top \mathrm{diag}(\mathbf{V}^{(-)}) \mathbf{D}$ where
\[
\mathbf{D} = \begin{bmatrix}
\mathbf{D}_1 & 0 \\
0 & \mathbf{D}_2 
\end{bmatrix},
\]
and $\mathbf{V}^{(-)} = [\mathbf{V}^{(1)}, \mathbf{V}^{(2)}]^\top$, where $V_i^{(1)} = E_{q(\mathbf{V})}[V_{1,i}^{-1}]$, and likewise for $\mathbf{V}^{(2)}$. The solution to the variational problem is:
\begin{align*}
q(\mathbf{x}, \boldsymbol{\theta}) &\sim  \mathrm{pLGM}  \{ \pi(\mathbf{y|\mathbf{x},\boldsymbol{\theta}_1}), \  \mathbf{m}=\mathbf{0}, \  \mathbf{Q}  = \mathbf{D}(\boldsymbol{\theta}_2)^\top \mathrm{diag} (\mathbf{V}^{\circ -1}) \mathbf{D}(\boldsymbol{\theta}_2), \  \pi(\boldsymbol{\theta}) \}, \\ 
 q(V_{1,i}) &\sim \mathrm{GIG}\left( -1, \ E_{q(\eta_1)}(\eta_1^{-1}), \ E_{q(\mathbf{x}_1, \theta)}([\mathbf{D}_1\mathbf{x}_1]_i^2]) + h_{1,i}^2E_{q(\eta_1)}(\eta_1^{-1}) \right), \ \ i=1,\dotsc,N_1, \\
  q(V_{2,i}) &\sim \mathrm{GIG}\left( -1, \ E_{q(\eta_2)}(\eta_2^{-1}), \ E_{q(\mathbf{x}_2, \theta)}([\mathbf{D}_2\mathbf{x}_2]_i^2]) + h_{2,i}^2E_{q(\eta_2)}(\eta_2^{-1}) \right), \ \ i=1,\dotsc,N_2, \\
 q(\eta_1) &\sim \mathrm{GIG}\left( -N_1/2 + 1, \ 2\alpha_{\eta_1}, \ \sum_{i=1}^N E_{q(V_{1,i})}(V_{1,i}) -2h_{1,i} + h_{1,i}^2 E_{q(V_{1,i})}(V_{1,i}^{-1})  \right), \\
 q(\eta_2) &\sim \mathrm{GIG}\left( -N_2/2 + 1, \ 2\alpha_{\eta_2}, \ \sum_{i=1}^N E_{q(V_{2,i})}(V_{2,i}) -2h_{2,i} + h_{2,i}^2 E_{q(V_{2,i})}(V_{2,i}^{-1})  \right).
\end{align*}

So we essentially need to repeat steps 6 to 11 of the CAVI Algorithm \ref{alg:MFVI1} for each set of mixing variables $\mathbf{V}_i$. For the SCVI algorithm, we consider the surrogate ${q(\mathbf{x}_1,\mathbf{x}_2,\boldsymbol{\theta},\mathbf{V}_1,\mathbf{V}_2) = q(\mathbf{x}_1,\mathbf{x}_2,\boldsymbol{\theta})q(\mathbf{V}_1)q(\mathbf{V}_2)}$ and, similarly, steps 5 and 6 in the CAVI \mbox{Algorithm} \ref{alg:MFVI2} need to be repeated for each set of mixing variables $\mathbf{V}_i$.




\subsection{Alternative driving noise distributions}

In sections \ref{sect:preliminaries} and \ref{sect:VIresult}, we restricted the driving noise distribution to the NIG distribution. We could have considered, however, other members of the generalized hyperbolic family, which we list in Table \ref{table:GH}. For example, the parameter $\eta$ of the GAL (generalized asymmetric Laplace) distribution has a similar interpretation as the NIG distribution (Gaussian model corresponds to the limiting case $\eta \to 0$). Still, in the $t$-distribution, it is the degrees of freedom parameter. As explained in \cite{wallin2015geostatistical}, when we define non-Gaussian processes in continuous space, the mixing distributions should be closed under convolution. However, the $t$-Student distribution is not closed under convolution; therefore, we should only consider this distribution for the driving noise for models defined in discrete space, where $h_i=1$.

\begin{table}
\begin{tabular}{ccc}
Distribution of $\Lambda_i$ & Mixing distribution of $V_i$ & GIG form of the mixing distribution \\
\hline$t$-Student & $\operatorname{IGamma}(\eta / 2, \eta / 2)$ & $\mathrm{GIG}(\eta / 2, 0, \eta)$ \\
$\mathrm{NIG}$ & $\mathrm{IGaussian}(h_i, h_i^2\eta^{-1})$ & $\mathrm{GIG}\left(-\frac{1}{2}, \eta^{-1}, h_i^2\eta^{-1} \right)$ \\
GAL & $\mathrm{Gamma}(h_i \eta^{-1}, \eta^{-1} )$ & $\mathrm{GIG}(h_i \eta^{-1}, 2 / \eta^{-1}, 0)$ \\
\end{tabular}
\caption{Special cases of the GH distribution, their mixing distribution, and GIG form.}
\label{table:GH}
\end{table}

The surrogate distribution $q(\mathbf{x},\boldsymbol{\theta})$ is still the same as in Theorem \ref{theo:1}. For the mixing distributions shown in Table \ref{table:GH}, all mixing variables $V_i$ have a GIG prior with some parameters $p_i,a_i,b_i$ and the surrogate $q(V_i)$ takes the form
\[
q(V_i) \sim \pi_{\mathrm{GIG}}\left( p_i-1/2, \ E_{q(\eta)}(a_i), \ E_{q(\mathbf{x}, \boldsymbol{\theta})}([\mathbf{D}\mathbf{x}]_i^2]) + E_{q(\eta)}(b_i) \right).
\]
For GAL driving noise, we consider an exponential prior for $\eta$ with some rate parameter $\alpha_\eta$. The surrogate distribution $q(\eta)$ is \\
\[
q(\eta) \propto \pi_{\mathrm{GIG}}(\eta; \  -(\textstyle \sum_i h_i)/\eta, \  2(\alpha_\eta + \textstyle \sum_i h_i E(\log(V_i))), \ 2 \textstyle \sum_i E(V_i) \ ) \prod_{i=1}^N \Gamma(h_i/\eta)^{-1},
\]
where the expectations are with respect to $q(\mathbf{V})$, and $\Gamma$ is the Gamma function. For $t$-Student driving noise, the surrogate distribution $q(\eta)$ is \\
\[
q(\eta) \propto \pi(\eta)\pi_{\mathrm{GIG}}(\eta; \  N\eta/2+1, \  N\log(2) + \textstyle \sum_i E(V_i^{-1}) + E(\log(V_i)), \ 0) \prod_{i=1}^N \Gamma(\eta/2)^{-1},
\] 
where $\pi(\eta)$ is the prior for $\eta$. We can sample from these distributions by computing their quantile functions numerically. It is also possible to consider skewed and long-tailed members of the GH distribution, but we did not consider these distributions in this paper.

\newpage
\section{Gibbs sampler}\label{sect:Gibbs}

The built-in Gibbs sampler implemented in the $\verb|ngvb|$ package iterates between the full conditionals:

\begin{align*}
\mathbf{x}, \boldsymbol{\theta}|\mathbf{V},\eta &\sim  \mathrm{pLGM}  \{ \pi(\mathbf{y}|\mathbf{x},\boldsymbol{\theta}_1), \  \mathbf{m}=\mathbf{0}, \  \mathbf{Q}  = \mathbf{D}(\boldsymbol{\theta}_2)^\top \mathrm{diag} (\mathbf{V}) \mathbf{D}(\boldsymbol{\theta}_2), \  \pi(\boldsymbol{\theta}) \}, \\ 
V_i|\mathbf{x}, \boldsymbol{\theta},\eta  &\sim \mathrm{GIG}\left( -1, \ \eta^{-1}, \ [\mathbf{D}(\boldsymbol{\theta}_2)\mathbf{x}]_i^2 + h_i^2\eta^{-1} \right), \ \ i=1,\dotsc,N, \\
\eta | \mathbf{x}, \boldsymbol{\theta}, \mathbf{V} &\sim \mathrm{GIG}\left( -N/2 + 1, \ 2\alpha_\eta, \ \sum_{i=1}^N V_i -2h_i + h_i^2 V_i^{-1}  \right). 
\end{align*}

We highlight the resemblance between these full conditionals to the solution of the variational problem in Theorem \ref{theo:1}. The main difference between the Gibbs sampler and the CAVI Algorithm \ref{alg:MFVI1} is that when we iterate between the full conditionals, we do not take the expectation of $\tilde{d_i} = [\mathbf{D}(\boldsymbol{\theta}_2)\mathbf{x}]_i^2, i= 1, \dotsc, N,$ nor the mixing variables $\mathbf{V}$ and parameter $\eta$. Instead, we generate samples from $\tilde{d_i}$, $\mathbf{V}$, and $\eta$.

\end{appendices}

\end{document}